\documentclass[12pt, draftclsnofoot, onecolumn]{IEEEtran}
\pdfoutput=1
 \usepackage{cite}

\ifCLASSINFOpdf
   \usepackage[pdftex]{graphicx}
   \graphicspath{{img/pdf/}{img/jpeg/}}
   \DeclareGraphicsExtensions{.pdf,.jpeg,.png}
\else
   \usepackage[dvips]{graphicx}
   \graphicspath{{img/eps/}}
   \DeclareGraphicsExtensions{.eps}
\fi

\usepackage[cmex10]{amsmath}
\usepackage{booktabs}
\usepackage{amsfonts}
\usepackage{amssymb}
\usepackage{amsthm}
\usepackage{mathrsfs}
\usepackage{bigints}
\usepackage{mathtools}

\usepackage[tight,footnotesize]{subfigure}
\usepackage{wasysym}
\usepackage{color}
\usepackage{epsfig} 
\usepackage{epstopdf}
\usepackage{float}
\usepackage{amsthm}
\usepackage{mathtools}
\usepackage{relsize}
\definecolor{light-gray}{gray}{0.8}
\def\nb0{{\mathbf{0}}}
\def\nb1{{\mathbf{1}}}

\newtheorem{lemma}{Lemma}
\newtheorem{thm}{Theorem}

\newtheorem{theorem}{Theorem}

\newtheorem{cor}{Corollary}

\newtheorem{remark}{Remark}

\def\E{\mathbb{E}}

\def\P{\mathbb{P}}

\def\R{\mathbb{R}}

\def\sir{\mathtt{SIR}}

\newcommand{\ud}{\, \mathrm{d}}

\newtheorem{defn}[thm]{Definition} 
\newtheorem{corollary}[cor]{Corollary}
\newcommand{\overbar}[1]{\mkern 1.5mu\overline{\mkern-1.5mu#1\mkern-1.5mu}\mkern 1.5mu}

\begin{document}
\title{\Huge Spatio-temporal Interference Correlation and Joint Coverage in Cellular Networks}
\author{\IEEEauthorblockN{Shankar Krishnan and Harpreet S. Dhillon \vspace{-5ex}}
\thanks{The authors are with Wireless@VT, Department of ECE, Virginia Tech, Blacksburg, VA, USA. Email: \{kshank93, hdhillon\}@vt.edu. The support of the US NSF (Grant CCF-1464293) is gratefully acknowledged. \hfill Last updated: \today} }
\maketitle
\begin{abstract}
This paper provides an analytical framework with foundations in stochastic geometry to characterize the spatio-temporal interference correlation as well as the joint coverage probability at two spatial locations in a cellular network. In particular, modeling the locations of cellular base stations (BSs) as a Poisson Point Process (PPP), we study interference correlation at two spatial locations $\ell_1$ and $\ell_2$ separated by a distance $v$, when the user follows \emph{closest BS association policy} at both spatial locations and moves from $\ell_1$ to $\ell_2$. With this user displacement, two scenarios can occur: i) the user is handed off to a new serving BS at $\ell_2$, or ii) no handoff occurs and the user is served by the same BS at both locations. After providing intermediate results such as probability of handoff and distance distributions of the serving BS at the two user locations, we use them to derive exact expressions for spatio-temporal interference correlation coefficient and joint coverage probability for any distance separation $v$. We also study two different handoff strategies: i) \emph{handoff skipping}, and ii) \emph{conventional handoffs}, and derive the expressions of joint coverage probability for both strategies. The exact analysis is not straightforward and involves a careful treatment of the neighborhood of the two spatial locations and the resulting handoff scenarios. To provide analytical insights, we also provide easy-to-use expressions for two special cases: i) static user ($v =0$) and ii) highly mobile user ($v \rightarrow \infty)$. As expected, our analysis shows that the interference correlation and joint coverage probability decrease with increasing $v$, with $v \rightarrow \infty$ corresponding to a completely uncorrelated scenario. Further design insights are also provided by studying the effect of few network/channel parameters such as BS density and path loss on the interference correlation.
\end{abstract}
\IEEEpeerreviewmaketitle
\allowdisplaybreaks
\section{Introduction}
Stochastic geometry has recently emerged as a popular tool for the modeling and analysis of large-scale wireless communication systems, such as wireless ad hoc and cellular networks~\cite{HaeB2013, Baccelli}. Irrespective of the type of wireless network being considered, the main idea behind these analyses is to model the locations of both the transmitters (Txs) and receivers (Rx) as point processes, often {\em independent} homogeneous Poisson point processes, and study {\em first-order} performance metrics, such as coverage probability, average rate, or interference distribution, as observed by a randomly chosen Rx, termed typical Rx \cite{HaeB2013}. While this approach of confining the analysis to a single observation point as well as almost always a single time-frequency resource {\em slice} is useful to get first-order insights, it is not sufficient for the characterization of spatio-temporal dependence (correlation) in the performance of a randomly chosen user in a wireless network. This requires the joint analysis of the observations made at two different spatial locations and/or different time-frequency resource {\em slices} which is known to be significantly more challenging than the more popular approach described above. As discussed next in detail, while this dependence has already been characterized for ad hoc networks, the same is not true for cellular networks for which the analysis is known to be far more challenging. Developing new tools to facilitate the exact characterization of spatio-temporal interference correlation as well as joint coverage probability in cellular networks is hence the main goal of this paper.
\subsection{Related work and Motivation} 
Correlation in interference observed at different locations or the same location at different times has been studied extensively for wireless networks, {\em albeit} almost exclusively in the context of wireless \emph{ad hoc} networks. In this ad hoc network setup, wireless nodes are usually assumed to transmit independently according to a random access scheme, such as ALOHA, with a certain transmit probability \cite{HaeB2013, Baccelli}. Interference is spatially correlated as it originates from the same set of transmitters even in the presence of independent fading. Similarly, interference is also temporally correlated since a subset of the same set of nodes transmit in different time slots.  The authors in \cite{Ganti} first characterized this interference correlation in ad hoc networks in terms of spatio-temporal correlation coefficient and showed that it is directly proportional to the random access probability and inversely proportional to the second moment of fading power gain. The authors in \cite{RelayCorrelation} extended this work and derived the joint probabilities of successful packet delivery at multiple receivers under spatially and temporally dependent interference. They showed that interference correlation significantly impacts the packet delivery probability. Along the same lines, \cite{SIMO} investigated interference correlation in multi-antenna receivers and showed that a diversity loss occurs due to interference correlation. Other related works include studying the network performance (in terms of outage probability and diversity order) for a decode-and-forward, cooperative relaying system \cite{Ralph1, CoopRelaying, Cooperative, CoopStrategies} under spatially and temporally correlated interference. The effect of interference correlation on the performance of multi-antenna communication systems under Maximal Ratio Combining is also studied recently in \cite{MRC1,MRC2}.

For the rest of this discussion, we note that the prior art on interference correlation can be broadly classified into two categories. The first line of work deals with the study of temporal correlation in \emph{static networks}, where nodes are static or have low mobility. Most of the works discussed above fall in this category. Additionally, considering the temporal correlation of interference over different time slots, \cite{DiversityPoly} derived closed form expressions for joint outage/success probability over $n$ time slots (transmissions) and showed that temporal diversity gain due to retransmissions diminishes with high interference correlation.  Along the same lines, \cite{BlockFading} investigated temporal correlation in the interference power for a correlated fading (flat fading, Rayleigh block fading) and correlated user traffic (slotted ALOHA) scenario. A key step in these studies is the characterization of moments of conditional success probability (conditioned on the point process), which have also been used recently for the derivation of {\em meta distribution} of the signal to interference ratio in both ad hoc and cellular networks in \cite{LowMobility,metadistribution}.
%


The second line of work deals with mobile networks, where mobility of nodes introduces randomness and thereby decorrelates interference across space and time. If the nodes are considered highly mobile ($v\rightarrow \infty$) in each time slot, the set of possible interferers also change rapidly and thus the interference becomes completely uncorrelated over time. However in real life scenarios, the nodes in a wireless network have finite mobility \cite{TVT} and it is thus important to study the interference correlation in finite mobile networks. For a mobile ad hoc network, the authors in \cite{Gong} studied temporal correlation in interference and outage under different mobility models such as Random Way point, Random Walk and Discrete-time Brownian motion. Specifically, they characterized correlation coefficient for interference and showed that correlation decreases with the increase in the mean speed of the nodes. More recently, \cite{BPPCorrelation} captured the effect of mobility on interference and outage correlation in finite networks (network with finite boundaries) and showed that interference correlation is location-dependent, being higher close to the network edge. Capturing the correlation in user locations in wireless networks, the authors in \cite{Cluster} studied interference correlation in clustered networks (modeled as Matern or Thomas cluster process) and showed that clustering of interferer locations enhances interference correlation.

Although there has been substantial work quantifying interference correlation in wireless ad hoc networks, there has been limited work studying correlation in cellular networks. Taking into account the temporal/spatial correlation in cellular networks, \cite{ICIC} analyzed the benefits of inter-cell interference coordination (ICIC) with BS coordination and intra-cell diversity (ICD) with selection combining in cellular networks. Their analysis of ICD showed that a diversity gain can always be obtained in a cellular setting with strongest BS association, in contrast with the conclusion drawn from ad hoc networks in \cite{DiversityPoly}. Studying spatiotemporal cooperation among BSs in a heterogeneous cellular network, \cite{Minero} studied different diversity exploiting techniques such as joint transmission and base station silencing. The authors in \cite{CorrelationHCN} studied the correlation of interference and link successes in heterogeneous cellular network with multiple tiers (different transmit powers and BS densities) and different cell association policies. However their analysis is not accurate as the authors also include the received power from the serving BS in the interference analysis. In this paper, we provide the exact analysis of interference correlation in cellular networks.
 
Before we describe our contributions in the next subsection, we provide a brief overview of the key differences in the analysis of ad hoc and cellular networks and explain why cellular analysis is more challenging. Ad hoc networks are usually modeled as Poisson bipolar networks \cite{Bipolar} with fixed link distance between the transmitter and the node of interest (receiver). The interference field can therefore be modeled as an infinite homogeneous PPP. However, in cellular networks, the serving BS has to be chosen based on some cell association strategy, such as maximum average received power based \cite{MaxRcvdPower} or highest instantaneous SINR based \cite{KTier}. This has an important implication on the modeling of interference field. In order to fix this key idea, let us consider {\em maximum average received power based cell association} in a single-tier cellular network in which the typical user connects to its closest BS. For this association strategy, the interference comprises of all active BSs farther than the closest BS (serving BS) and hence the location of this serving BS plays an important role in the analysis of any performance metric that depends upon the received power of the desired signal and/or the interfering signal. Moreover, to study interference correlation in two spatial locations (a finite distance apart), it is important to characterize the distance of the serving BS at both spatial locations. Also note that there is a certain probability that the serving BS at the second location is the same as the one in the first location. This dependence among serving BSs in cellular networks is characterized by the \emph{handoff rate} \cite{Handoff}, which is the probability that the user is handed off to a new serving BS as it moves from one location to another. This characterization of handoff is not required in ad hoc networks. In this work, we present exact characterization of both interference correlation as well as the joint coverage probability while incorporating all such dependencies. To the best of our knowledge, we are first ones to characterize these correlation-based metrics for cellular networks.

\subsection{Contributions and Outcomes} 
\subsubsection*{Distance distribution of serving BS at two spatial locations}
Incorporating spatio-temporal corrrelation, we study the network performance when a typical user moves a distance $v$ from its initial location in a cellular network.  The user follows closest BS association i.e. connects to the closest BS at both user locations. We first identify that the user displacement can result in two scenarios: i) \emph{No Handoff}: the user is associated with the same BS at both the user locations (same BS is the closest at both locations), and ii) \emph{Handoff}: A different BS is closer to the user at the second location and therefore handoff occurs. For both these scenarios, we provide joint distribution of the distances from the two locations to their respective closest BSs. 
\subsubsection*{Spatio-temporal interference correlation coefficient for cellular networks}
After characterizing distance distributions, we first study the spatio-temporal interference correlation coefficient of a typical user in a cellular network under closest BS association policy by deriving expressions for mean, variance and first order cross moments of interference at two user locations in the network. We then show that interference correlation decreases with the distance $v$ between the two spatial locations, becoming uncorrelated at locations far apart. As a special case, we derive an easy-to-use expression for the temporal interference correlation coefficient, i.e., when the user is static at a given spatial location. The temporal correlation coefficient in cellular networks is shown to exhibit a \emph{bell curve} relationship w.r.t. BS density and decrease with path loss for large $v$. This \emph{bell curve trend} does not exist in the case of ad hoc networks.
\subsubsection*{Exact analysis of joint coverage probability at two spatial locations}
Joint coverage probability at two spatial locations completely characterizes the dependence in link successes (coverage) at the two locations in the network. We develop exact expressions for the joint coverage probability for two spatial locations (separated by a distance $v$) under two handoff strategies: i) \emph{handoff skipping} and ii) \emph{conventional handoffs}. With handoff skipping, a user initially connected to its closest BS at a given location continues to be associated with the same BS irrespective of the distance it moves and skips all possible handoffs to a new serving BS. Handoff skipping is more relevant to ultra-dense networks, where a user can save handoff delays/overheads by skipping certain handoffs. Next, we look at the conventional handoff scenario where a user always switches its association to the closest BS as it moves along in the network. For both handoff strategies, we show that joint coverage probability at two spatial locations decreases with the separation in the two locations (i.e. decrease in correlation).
\section{System Model} \label{sec2}

We consider a cellular network where the base stations are modeled as a homogeneous PPP $\Phi$ of intensity $\lambda$. As noted already, our main objective is to study how the interference experienced by a typical user is correlated across two spatial locations in this network. At any spatial location in the network, the typical user follows the \emph{closest BS association policy} i.e. connects to its closest BS as that maximizes its average received power. Also, we study how the interference is correlated temporally as the typical user is subject to interference from the same subset of BSs across time.
\begin{figure}[t!]
\centering{
\includegraphics[width=.90\linewidth]{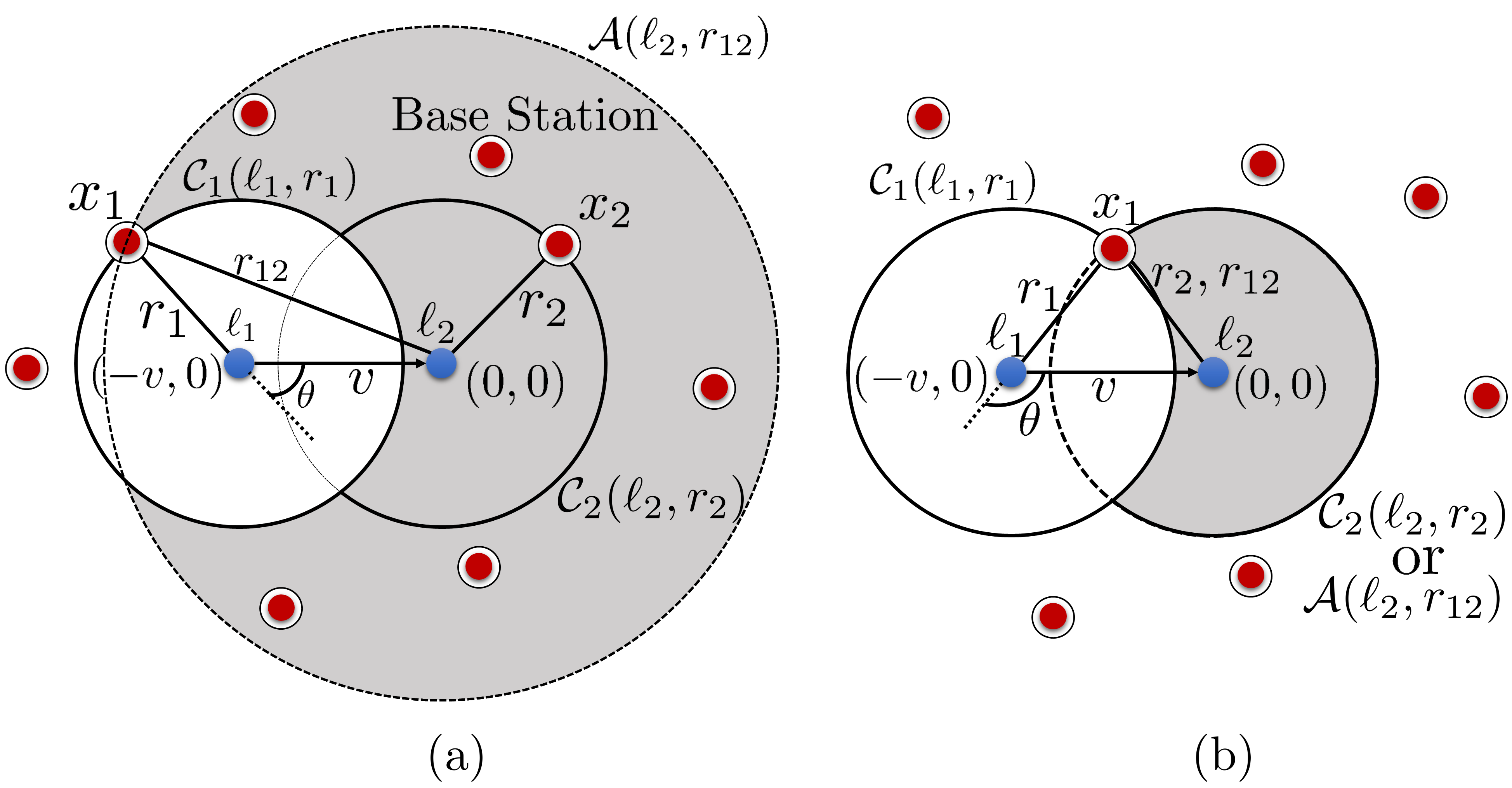}
\caption{System model when a typical user (denoted by \emph{blue dot}) moves from location 1 ($\ell_1$) to location 2 ($\ell_2$). (a) \emph{Handoff scenario}: user is handed off to a new serving BS $x_2$ at $\ell_2$ as it is closer than the serving BS $x_1$ at $\ell_1$, and (b) \emph{No Handoff scenario}: user is served by the same BS $x_1$ at both locations $\ell_1$ and $\ell_2$ and therby no handoffs occur.}
\label{Fig:SystemModel}
}
\end{figure}
\subsection{System Setup and Key assumptions}
Consider a typical user that connects to its closest BS $x_1 \in \Phi$ at location 1 ($\ell_1$) as shown in Fig. \ref{Fig:SystemModel}. The user now moves a distance $v$ at an angle $\theta$ to the serving BS $x_1$ at $\ell_1$  and shifts to location 2 ($\ell_2$). With this user displacement, two scenarios can arise : i) a different BS $x_2 \in \Phi$ is located closer to the user at location 2 than the serving BS $x_1$ at location 1 (Fig. \ref{Fig:SystemModel} (a)) or ii) the BS $x_1$ (serving BS at location 1) is still the closest BS to the user at location 2 (Fig. \ref{Fig:SystemModel} (b)). The first scenario corresponds to the case where handoff occurs to a new serving BS $x_2$ while the second scenario results in no handoff and a continued connection to the BS $x_1$. In this work, we study how the interference is correlated across these two user locations $\ell_1$ and $\ell_2$. Before that, we need to define the distance distributions of the serving BS at the two spatial locations and understand the above mentioned scenarios (handoff and no handoff) in a bit more detail.

Let $R_1$ and $R_2$ be the random variables denoting the distance of the serving (closest) BS to the user at locations 1 and 2 respectively, with $r_1$ and $r_2$ being their realizations. As can be seen from Fig. \ref{Fig:SystemModel}, the user's location 2 is at a distance $r_{12} = \sqrt{r_1^2+v^2+2r_1v\cos \theta}$ (using law of cosines) from BS $x_1$. Let $\Theta$ be a uniform random variable in $[0,\pi]$ denoting the angle of user displacement as the user moves from $\ell_1$ to $\ell_2$ with $\theta$ being its realization. Therefore its PDF is given as $f_{\Theta}(\theta) = 1/\pi$. Denote by $\mathcal{C}_1(\ell_1,r_1)$ - a circle centered at $\ell_1$ with radius $r_1$ and  $\mathcal{C}_2(\ell_2,r_2)$ - a circle centered at $\ell_2$ with radius $r_2$. And let $\mathcal{A}(\ell_2,r_{12})$ be a circle centered at $\ell_2$ with a radius of $r_{12}$, which will be used for the handoff analysis. Before that, we state an intermediate result which will be useful in our analysis.
\begin{defn} \label{Definition 1}
Consider two intersecting circles $\mathcal{C}_1(\ell_1,r_1)$ and  $\mathcal{A}(\ell_2,r_{12})$ with centers separated by a distance $v$ and an angle $\theta$ as shown in Fig. \ref{Fig:SystemModel}. The shaded region in the figure is denoted by $\mathbf{C}(\ell_1,r_1,\ell_2,r_{12},v)$ and its area given as
\begin{align}
|\mathbf{C}(\ell_1,r_1,\ell_2,r_{12},v)| = r_{12}^2\bigg[\pi - \theta + \sin^{-1}\bigg(\frac{v\sin\theta}{r_{12}}\bigg)\bigg]  - r_1^2(\pi - \theta) + r_1v\sin\theta.
\end{align}
\begin{proof}
The area of the shaded region is 
\begin{align}
\notag |\mathbf{C}(\ell_1,r_1,\ell_2,r_{12},v)| &= |\mathcal{A}(\ell_2,r_{12}) \setminus \mathcal{A}(\ell_2,r_{12}) \cap \mathcal{C}_1(\ell_1,r_1)| \\\notag
&= \pi r_{12}^2 - |\mathcal{A}(\ell_2,r_{12}) \cap \mathcal{C}_1(\ell_1,r_1)|
\end{align}
where $|.|$ denotes the area, and the result follows by using the area of intersection of two circles as done in ~\cite[Theorem 1]{Handoff}.
\end{proof}
\end{defn}
As the user moves from $\ell_1$ to $\ell_2$, whether handoff occurs (Fig. \ref{Fig:SystemModel}(a)) or not (Fig. \ref{Fig:SystemModel}(b)) is dictated by the existence of a BS within the circle $\mathcal{A}(\ell_2,r_{12})$. Let $H$ be the event that handoff occurs as the user moves from $\ell_1$ to $\ell_2$ and $\bar{H}$ be the complementary event (no handoffs occur). The probability of handoff $\P(H)$ is derived in ~\cite[Theorem 1]{Handoff} for a similar setup and is restated below in Lemma \ref{Lemma 1} for completeness. We use a shorthand notation $\mathcal{C}_1$, $\mathcal{C}_2$ and $\mathcal{A}$ for denoting the circles defined before for the sake of simplicity.
\begin{lemma} \label{Lemma 1}
Conditioned on $r_1$ and $\theta$, the probability of handoff as the user moves a distance $v$ at an angle $\theta$ from location 1 to 2 in a PPP of BS density $\lambda$ is
\begin{align}
\P(H|r_1,\theta) &= 1- \exp\bigg( - \lambda \bigg(r_{12}^2\bigg[\pi - \theta + \sin^{-1}\bigg(\frac{v\sin\theta}{r_{12}}\bigg)\bigg] - r_1^2(\pi - \theta) + r_1v\sin\theta\bigg)\bigg). \label{Eq:Prob handoff}
\end{align}
\begin{proof}
From Fig. \ref{Fig:SystemModel}, for a typical user initially connected to its closest BS $x_1$ at distance $r_1$ and moving to a new location $\ell_2$ at distance $r_{12}$ from BS $x_1$, a handoff does not occur if there is no BS closer than $r_{12}$ to the user at $\ell_2$, hence:
\begin{align}
\notag 1- \P(H|r_1,\theta) & \stackrel{(a)}= \P(N(|\mathcal{A}|) = 1 | N(|\mathcal{A} \cap \mathcal{C}_1|) = 1) \\\notag
&= \P(N(|\mathcal{A} \setminus \mathcal{A} \cap \mathcal{C}_1|) = 0) \\\notag
&= \exp(-\lambda \: |\mathbf{C}(\ell_1,r_1,\ell_2,r_{12},v)|)
\end{align} 
where (a) follows because only one BS $x_1$ lies in the region $\mathcal{A}$ for a handoff scenario and the result follows by using $|\mathbf{C}(\ell_1,r_1,\ell_2,r_{12},v)|$ from Definition \ref{Definition 1}.
\end{proof}
\end{lemma}
We now state some observations about the system model depending on the user displacement $v$ in Remarks 1 and 2.
{\remark \label{Remark 1}
Conditioned on the occurence of handoff (Fig. \ref{Fig:SystemModel} (a)) when a user moves a distance $v$ from $\ell_1$ to $\ell_2$ , one of the following 3 cases arises for circles $\mathcal{C}_1$ and $\mathcal{C}_2$ : (i) {Case 1:} Disjoint circles ($v \geqslant r_1+r_2$), (ii) {Case 2:} Intersecting circles ($ r_2 - r_1 < v < r_1 +r_2 $), and (iii) {Case 3:} Engulfed circles ($v \leqslant r_2 - r_1$). This insight will be useful for the analysis.}
\begin{figure}[t!]
\centering{
\includegraphics[width=.90\linewidth]{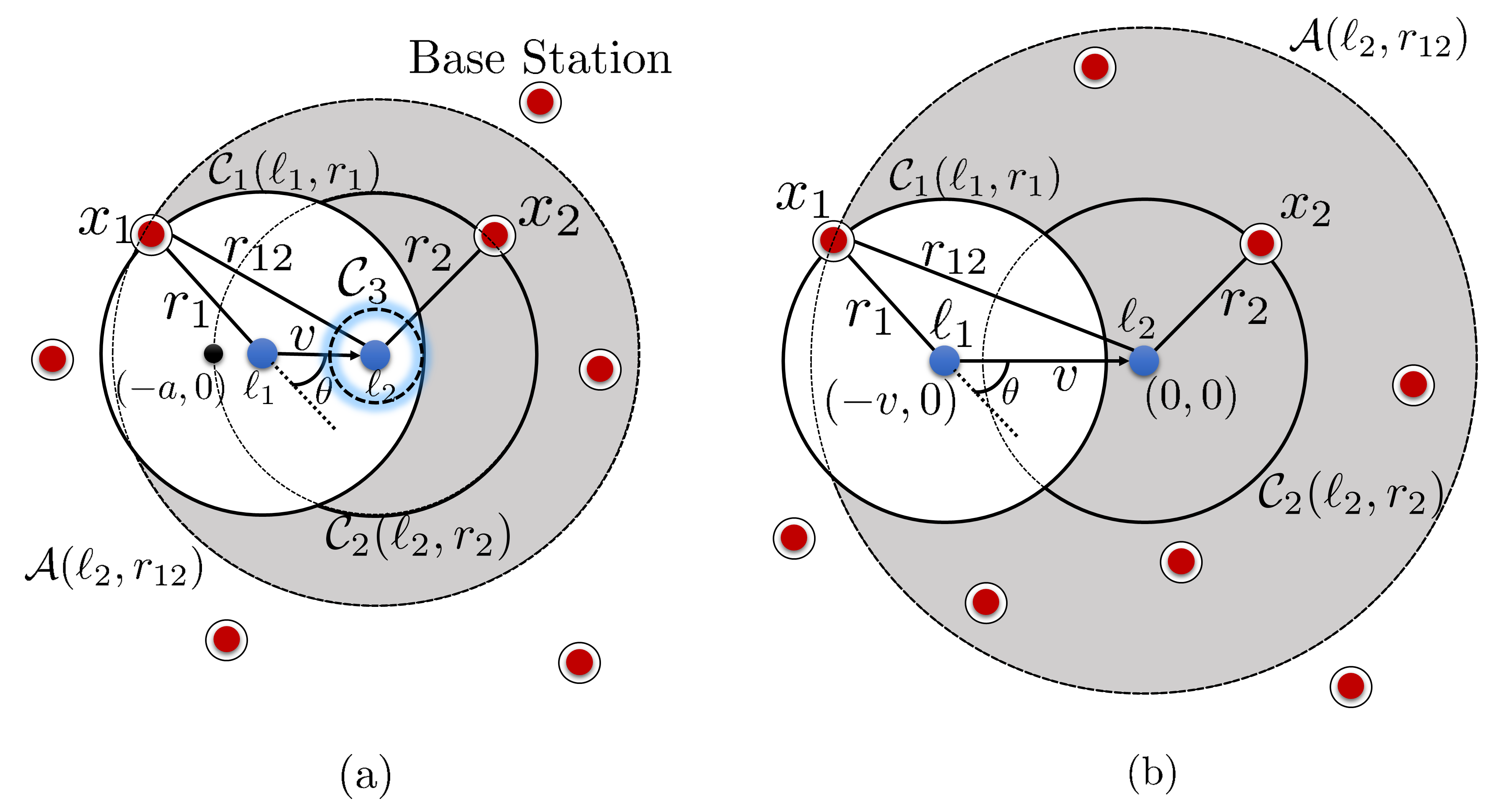}
\caption{Different scenarios based on user displacement $v$ from $\ell_1$ to $\ell_2$. (a) Scenario 1 ($v<r_1$), where $\mathcal{C}_3(\ell_2,r_1-v)$ is an exclusion zone and (b) Scenario 2 ($v \geq r_1$), where there is no exclusion zone.}
\label{Fig:Scenarios}
}
\end{figure}
{\remark
For the system model shown in Fig. \ref{Fig:SystemModel}, there exists two scenarios based on the distance $v$ moved by the user from $\ell_1$ to $\ell_2$: i) Scenario 1 ($v < r_1$) and ii) Scenario 2 ($v \geq r_1$) (see Fig. \ref{Fig:Scenarios}). For scenario 1, even after user displacement, the user is still inside $\mathcal{C}_1$ and hence no BS lies can lie within a distance $r_1-v$ from $\ell_2$ i.e. the closest BS is atleast a distance of $r_1-v$ from $\ell_2$ or $r_2 \geq r_1 -v$. In scenario 2, the user moves a larger distance ($v \geq r_1$) and no such condition exist for the serving distance $R_2$ i.e. $r_2 \geq 0$. We define $z_1 = \max(0,r_1-v)$   and circle $\mathcal{C}_3(\ell_2,z_1)$ to handle the two scenarios together, which will be discussed later.\\
}

As stated before, $R_1$ is the distance of the closest BS from location 1 (closest point of PPP $\Phi$ from $\ell_1$). Hence the distribution of $R_1$ is given by the null probability of the PPP $\Phi$ and is thus given as $f_{R_1} (r_1) = 2\lambda \pi r_1 e^{-\lambda \pi r_1^2}$ \cite{HaeB2013}. The distribution of $R_2$, the distance of the serving BS at location 2 depends whether a handoff occurs or not when user moves from $\ell_1$ to $\ell_2$. If there is no handoff, the serving BS at location 2 is same as the one at location 1. In that scenario, $r_2 = r_{12} = \sqrt{r_1^2+v^2+2r_1v\cos \theta}$ as evident from Fig. \ref{Fig:SystemModel}(b) and its distribution can be obtained accordingly from the distributions of $R_1$ and $\Theta$. However when handoff occurs, the distribution of $R_2$ is not straightforward and is derived next in Lemma 2.
\begin{lemma} \label{Lemma 2}
Conditioned on $r_1, \theta$ and the occurrence of a handoff (event $H$) as the user moves a distance $v$ at an angle $\theta$ from $\ell_1$ to $\ell_2$ in a PPP of intensity $\lambda$ as shown in Fig. \ref{Fig:SystemModel}, the CDF of distance $R_2$ of the serving BS at location 2 is given as
\begin{align*} \label{Eqn:R2 Handoff}
F_{R_2|H}(r_2|H,r_1,\theta) &= \left\{
     \begin{array}{lr}
       \frac{1-\exp(-\lambda \: |\mathbf{C}(\ell_1,r_1,\ell_2,r_2,v)|)}{1-\exp(-\lambda \: |\mathbf{C}(\ell_1,r_1,\ell_2,r_{12},v)|))},  \qquad r_2 \in [z_1,r_{12}] \\
        0, \qquad \qquad \qquad \qquad \qquad \qquad \: \: \:  \text{otherwise}\\
      \end{array}
   \right. ,
\end{align*}
where $z_1 = {\tt max}(0,r_1-v)$.
\end{lemma}
\begin{proof}
From Remark 2 and Fig. \ref{Fig:SystemModel}(a) (handoff scenario), it can be concurred that the distance $R_2$ between the new serving BS $x_2$ and the user at $\ell_2$ is greater than $z_1 = \max(0,r_1 - v)$. Also, it can not be farther than $r_{12}$ because otherwise it would not be closer than the serving BS $x_1$ at location 1. Hence $R_2 \in [\max(0,r_1-v),r_{12}]$. Conditioned on the occurence of handoff, the CDF of $R_2$ is thus given as:
\begin{align}
\notag \P(R_2 \leq r_2|H,r_1,\theta) & \stackrel{(a)}= \frac{1}{\P(H|r_1,\theta)}\big[1- \P(N\big(\mathbf{C}(\ell_1,r_1,\ell_2,r_{2},v)\big)=0)\big] \\
& = \frac{1}{\P(H|r_1,\theta)}\big[1- \exp(-\lambda \: |\mathbf{C}(\ell_1,r_1,\ell_2,r_{2},v)|)\big]
\end{align}
where (a) follows because conditioned on the presence of no BSs inside $\mathcal{C}_1$, the distribution of $R_2$ is dictated by the presence of no BSs in the region $\mathbf{C}(\ell_1,r_1,\ell_2,r_2,v) = |\mathcal{C}_2 \setminus \mathcal{C}_2 \cap \mathcal{C}_1|$ (same logic as Lemma \ref{Lemma 1}) and the final result follows by using $P(H|r_1,\theta)$ from Lemma \ref{Lemma 1}.
\end{proof}
\subsection{Channel Model}
We assume all BSs transmit with unit power and consider Rayleigh fading links with mean power gain of unity. We assume the fading gains across all links to be spatially and temporally independent and the fading coefficient between two nodes $x$ and $y$ at any time slot $k$ as $h_{xy}(k)$. We consider the following bounded path-loss function $g(x)$ for the large-scale fading, 
\begin{align}   \label{Eq:Path loss function}
g(x) = \frac{1}{\epsilon + \|x\|^{\alpha}} , \: \text{where} \: \epsilon > 0.   
\end{align} 
We next define the signal and interference experienced by a typical user at two spatial locations in two different time slots for the channel model described above. Say the typical user located at $\ell_1$ connects to its closest BS $x_1$ at time slot $t_1$. At a different time slot $t_2$, the user now located at $\ell_2$ connects to its closest BS $x_2$. Let $S(t_1,\ell_1)$ and $S(t_2,\ell_2)$ be the received powers from the serving BSs $x_1$ and $x_2$ at time slots $t_1$ and $t_2$  (user locations $\ell_1$ and $\ell_2$) respectively. Then 
\begin{align}
S(t_1,\ell_1) = h_{x_1\ell_1}(t_1)g(x_1-\ell_1) \label{Eq:RecPowloc1}\\ 
S(t_2,\ell_2) = h_{x_2\ell_2}(t_2)g(x_2-\ell_2). \label{Eq:RecPowloc2}
\end{align}
The interference power at time slots $t_1$ and $t_2$ with the user at locations $\ell_1$ and $\ell_2$ is denoted by $I(t_1,\ell_1)$ and $I(t_2,\ell_2)$ respectively, which are given as
\begin{align}
I(t_1,\ell_1) = \sum \limits_{x \in \Phi, x \neq x_1} h_{x\ell_1}(t_1) g(x-\ell_1) \label{Eq:Intloc1}\\
I(t_2,\ell_2) = \sum \limits_{x \in \Phi, x \neq x_2} h_{x\ell_2}(t_2) g(x-\ell_2). \label{Eq:Intloc2}
\end{align} 
The Signal-to-Interference Ratio ($\sir$) at time slot $t_1$ (user located at $\ell_1$) is denoted by $\sir(t_1,\ell_1)$ and is thus given as 
\begin{align} \label{Eq:SIR1}
\sir(t_1,\ell_1) = \frac{S(t_1,\ell_1)}{I(t_1,\ell_1)} =\frac{h_{x_1\ell_1}(t_1)g(x_1-\ell_1)}{\sum \limits_{x \in \Phi, x \neq x_1} h_{x\ell_1}(t_1) g(x-\ell_1)}.
\end{align}
Similarily, $\sir$ at time slot $t_2$ (user located at $\ell_2$) is given as
\begin{align} \label{Eq:SIR2}
\sir(t_2,\ell_2) = \frac{S(t_2,\ell_2)}{I(t_2,\ell_2)} =\frac{h_{x_2\ell_2}(t_2)g(x_2-\ell_2)}{\sum \limits_{x \in \Phi, x \neq x_2} h_{x\ell_2}(t_2) g(x-\ell_2) }.
\end{align}
For this system setup, we first study the interference correlation at two spatial locations $\ell_1$ and $\ell_2$ (time slots $t_1$ and $t_2$)  in terms of spatio-temporal correlation coefficient in Section \ref{sec:sec3} which is defined below:
\begin{align}
\zeta_I(\ell_1,\ell_2) = \frac{\E[I(t_1,\ell_1)I(t_2,\ell_2)] - \E[I(t_2,\ell_2)]^2}{\E[{I(t_2,\ell_2)^2}]-\E[{I(t_2,\ell_2)]^2}}. \label{Corr Coeff}
\end{align}
Then, we study correlation in link successes at the two user locations $\ell_1$ and $\ell_2$ in terms of joint coverage probability in Section \ref{sec:sec4}, which is formally defined next.
\begin{defn} (Joint Coverage Probability) \label{Def:def2} 
It is defined as the probability that the $\sir$ at user's spatial location $\ell_1$ (time slot $t_1$) and spatial location $\ell_2$ (time slot $t_2$) both exceed a certain SIR target (threshold) $T$. In this paper, it is denoted by ${\rm P_{c}}(\ell_1,\ell_2)$ and is mathematically defined as
\begin{align}
{\rm P_{c}}(\ell_1,\ell_2) = \P(\sir(t_1,\ell_1)>T,\sir(t_2,\ell_2)>T).
\end{align}
\end{defn}
\section{Spatio-temporal Interference Correlation} \label{sec:sec3}
The joint interference statistics at two spatial locations/time slots in a wireless network capture the effect of spatial/temporal correlation and is helpful in characterizing the network performance. However, deriving such joint statistics is usually not straightforward. The authors in \cite{JointInter} analyzed the joint temporal statistics of interference for a single-hop communication link by deriving the joint characteristic function of interference which follows a multivariate symmetric alpha
stable distribution. The expressions, though not obtained in closed form, facilitated the evaluation of different performance metrics such as  local delay, average network throughput and transmission capacity in the low-outage regime. In \cite{Gong}, the authors study outage correlation in mobile ad hoc networks and state that the direct evaluation of the joint distribution of correlated interference at two time slots is impractical and hence provide lower and upper bounds for the same. In this work, we derive the exact joint statistics of interference and coverage observed at two user locations. In this Section, we focus on the characterization of spatio-temporal interference correlation coefficient, whereas in Section Section \ref{sec:sec4} we will study joint coverage probability. 

In this work, we consider two spatial locations of the user at a distance $v$ apart in two different time slots, which w.l.o.g. are taken as time slot $1$ and $2$ i.e. $t_1 = 1$ and $t_2 = 2$. As shown in Fig. \ref{Fig:SystemModel}, the two user locations are taken to be $\ell_1 = (-v,0)$ and $\ell_2 = (0,0)$. We use the shorthand notation $I(1)$ and $I(2)$ to denote the interference at time slots $1$ and $2$ (spatial locations $\ell_1$ and $\ell_2$) and using (\ref{Eq:Intloc1}) and (\ref{Eq:Intloc2}), it is given as 
\begin{align}
 {I}(1) &= \sum \limits_{x \in \Phi, x \neq x_1} h_x(1) g({x}-v) \label{Eq:I1}
 \end{align}
 \begin{align}
{I}(2) &= \sum \limits_{x \in \Phi, x \neq x_2} h_x(2) g(x), \label{Eq:I2}
\end{align} 
where $h_x(1)$ and $h_x(2)$ denote the fading coefficients between a node $x \in \Phi$ and the user at locations $\ell_1$ and $\ell_2$ respectively.
Again for simplicity, a shorthand notation $S(1)$ and $S(2)$ is used to denote the received power from the serving BSs $x_1$ and $x_2$ at time slots $1$ and $2$ respectively. From  (\ref{Eq:RecPowloc1}) and (\ref{Eq:RecPowloc2}), $S(1) = h_{x_1}(1)g(x_1-v)$ and $S(2) = h_{x_2}(2)g(x_2)$.  Let the total received power (signal plus interference) at time slots $1$ and $2$ be denoted by ${I_{\rm t}(1)}$ and  ${I_{\rm t}(2)}$ and is  given as
\begin{align}
{I_{\rm t}(1)} &= \sum \limits_{x \in \Phi} h_x(1) g({x-v) = I(1)} + h_{x_1}(1) g(x_1-v)  \\ 
{I_{\rm t}(2)} &= \sum \limits_{x \in \Phi} h_x(2) g(x) = I(2) + h_{x_2}(2) g(x_2). \label{Adhoc Int}
\end{align}
The expression of total received powers  $I_{\rm t}(1)$ and $I_{\rm t}(2)$ for our setup is the same as the interference experienced by a typical user in an ad-hoc network with unit random access channel probability. The mean, second moment and first order cross moment of interference in such an ad-hoc network at time slots $1$ and $2$ is derived in \cite{Ganti} and will be useful in our analysis. The expressions are given below.
\begin{align}
& \E[I_{\rm t}(1)] = \E[I_{\rm t}(2)] = \lambda \int_{\R^2} g(x) dx \label{It1}\\
 &\E[{I_{\rm t}(1)}^2] = \E[{I_{\rm t}(2)}^2] = E[h^2] \lambda \int_{\R^2} g^2(x) dx +  \lambda^2 \left(\int_{\R^2} g(x) dx\right)^2 \label{It11}\\
&\E[{I_{\rm t}(1)}{I_{\rm t}(2)}] = \lambda \int\limits_{\R^2} g(x-v)g(x) dx + \lambda^2 \left(\int\limits_{\R^2} g(x) dx\right)^2 \label{It12}
\end{align}
We perform a similar analysis to determine the spatio-temporal correlation of interference in a cellular network with closest-BS association policy. By definition (from Equation (\ref{Corr Coeff})), the spatio-temporal interference correlation coefficient for a typical user at spatial locations $\ell_1$ and $\ell_2$ (at time slots $1$ and $2$) is hence given as
\begin{align}
\zeta_I(\ell_1,\ell_2) = \frac{\E[I(1)I(2)] - \E[I(2)]^2}{\E[{I(2)^2}]-\E[{I(2)]^2}}. \label{Corr Coeff new}
\end{align}
The mean, second moment and first order cross moment of interference for the typical user is derived next to evaluate the expression for spatio-temporal correlation coefficient given by (\ref{Corr Coeff new}).
The mean of interference at time slot $2$ is 
\begin{align}
 \notag \E[I(2)] & \stackrel{(a)}=  \E[I_{\rm t}(2)-h_{x_2}(2)g(x_2)] \\
& \stackrel{(b)}= \E[I_{\rm t}(2)]-\E[g(x_2)],  \label{Mean Int}
\end{align}
where (a) follows from the definition of $I_{\rm t}(2)$ in (\ref{Adhoc Int}) and (b) results as $\E[h] = 1$.
From (\ref{Mean Int}),
\begin{align}
\notag  \E[I(2)]^2 &= \E[I_{\rm t}(2)]^2+ \E[g(x_2)]^2 - 2\E[I_{\rm t}(2)]\E[g(x_2)] \\
& \stackrel{(c)}= \E[I_{\rm t}(2)]^2+ \E[g(x_2)]^2 -2\lambda \E[g(x_2)\int_{\R^2} g(x) dx ], \label{Mean Int Square}
\end{align}
where (c) follows by using (\ref{It1}) and the property of expectations i.e. $\E[cX] = c\E[X]$ where $c$ is a constant.

The second moment of the interference can be now computed as 
\begin{align}
\E[I(2)^2] &  \notag =  \E[(I_{\rm t}(2)-h_{x_2}(2)g(x_2))^2] \\\notag
 & = \E[{I_{\rm t}(2)}^2] + \E[h^2]\E[g^2(x_2)]-2\E[h_{x_2}(2)g(x_2)I_{\rm t}(2)] \\\notag
 & = \E[{I_{\rm t}(2)}^2]+\E[h^2]\E[g^2(x_2)]  -2\E[h_{x_2}(2)g(x_2)(h_{x_2}(2)g(x_2)+ {I}(2))] \\\notag
& \stackrel{(d)}= \E[{I_{\rm t}(2)}^2]-\E[h^2]\E[g^2({x_2})] - 2\E\big[h_{x_2}(2)g(x_2)\sum \limits_{x \in \Phi, x \neq x_2} h_x(2) g(x)\big] \\
& \stackrel{(e)}= \E[{I_{\rm t}(2)}^2]-\E[h^2]\E[g^2({x_2})] - 2\lambda\E\big[g(x_2)\int_{\R^2 \backslash \mathcal{C}_2}^{\infty} g(x) \ud x\big], \label{I11}
\end{align}
where (d) is obtained by using the expression of $I(2)$ in (\ref{Eq:I2}) and (e) follows from Campbell's law and the spatial independence of fading links. Here $\mathcal{C}_2$ denotes the circle centered at $l_2$ i.e. origin and radius $r_2$ as shown in Fig. \ref{Fig:SystemModel}.

Now, in order to evaluate correlation coefficient defined by (\ref{Corr Coeff new}), we are left to evaluate $E[I(1)I(2)]$, which we do next.  
\begin{align}
\notag & \E[I(1)I(2)]  = \E[(I_{\rm t}(1)-h_{x_1}(1)g(x_1-v))(I_{\rm t}(2)-h_{x_2}(2)g(x_2))] \\\notag
& = \E[I_{\rm t}(1)I_{\rm t}(2)]+\E[h_{x_1}(1)h_{x_1}(2)g(x_1-v)g(x_2)] -\E[h_{x_1}(1)g(x_1-v)I_{\rm t}(2)]-\E[h_{x_2}(2)g(x_2)I_{\rm t}(1)] \\\notag
& = \E[I_{\rm t}(1)I_{\rm t}(2)]+\E[g(x_1-v)g(x_2)]- \E[h_{x_1}(1)g(x_1-v)(h_{x_2}(2)g(x_2)+ {I}(2))] \\\notag & \qquad \qquad - \E[h_{x_2}(2)g(x_2)(h_{x_1}(1)g(x_1-v)+ {I}(1))]\\
& = \E[I_{\rm t}(1)I_{\rm t}(2)]-\E[g(x_1-v)g(x_2)]-\underbrace{\E[h_{x_1}(1)g(x_1-v)I(2)]}_{T_1}-\underbrace{\E[h_{x_2}(2)g(x_2)I(1)]}_{T_2} \label{I1I2}
\end{align}
The expressions of $T_1$ and $T_2$ are further simplified by proceeding as below.
\begin{align}
\notag T_1 &= \E[h_{x_1}(1)g(x_1-v)I(2)] \\\notag
&\stackrel{(f)}= \E[h_{x_1}(1)g(x_1-v)I(2),\bar{H}]+ \E[h_{x_1}(1)g(x_1-v)I(2),H] \\\notag
&\stackrel{(g)}=\E\big[h_{x_1}(1)g(x_1-v)\sum \limits_{x \in \Phi, x \neq \{x_1,x_2\}} h_x(2) g(x),\bar{H}\big]+ \\\notag & \qquad \qquad \E\big[h_{x_1}(1)g(x_1-v)\big(h_{x_1}(2)g(x_1) + \sum \limits_{x \in \Phi, x \neq \{x_1,x_2\}} h_x(2) g(x)\big),H\big] \\\notag
& = \E\big[g(x_1-v)\sum \limits_{x \in \Phi, x \neq \{x_1,x_2\}} g(x)\big]+ \E[g(x_1-v)g(x_1),H] \\
& \stackrel{(h)}= \lambda\E\bigg[g(x_1-v)\int_{\R^2 \setminus (\mathcal{C}_1 \cup \mathcal{C}_2)}g(x) \ud x\bigg] +\E[g(x_1-v)g(x_1),H], \label{T1}
\end{align}
where (f) follows by splitting the expectation into two possible scenarios (no handoff and handoff). This step is taken to consider the interference in the second time slot $I(2)$ appropriately. In case of handoff, interference from  BS $x_1$ also needs to be considered at $\ell_2$ whereas in case of no handoff, the BS $x_1$ continues to be the serving BS at location 2 and hence should not be considered as a part of $I(2)$. Step (g) follows from the above argument and considering the interference from $x_1$ only in the handoff scenario. Step (h) follows by applying Campbell's law and observing that interference excluding BSs $x_1$ and $x_2$ is equivalent to considering interference outside $\mathcal{C}_1$ and $\mathcal{C}_2$, where $\mathcal{C}_1$ and $\mathcal{C}_2$ are as shown in Fig. \ref{Fig:SystemModel}.

Proceeding similar to $T_1$, we obtain 
\begin{align}
T_2 = \lambda\E\bigg[g(x_2)\int_{\R^2 \setminus (\mathcal{C}_1 \cup \mathcal{C}_2)}g(x-v) \ud x\bigg]+\E[g(x_2)g(x_2-v),H]. \label{T2}
\end{align}

Now substituting various moment expressions given by (\ref{Mean Int Square}), (\ref{I11}), and (\ref{I1I2}) in (\ref{Corr Coeff new}), we get a general expression for the spatio-temporal correlation coefficient as a function of $v$. While it is not straightforward to gain analytical insights from the final expression (given its complexity), we will revisit this general case in the Numerical Results section (Section \ref{Sec:sec5}). In the rest of this Section, we focus on the more tractable case of $v=0$, which corresponds to the static user, i.e., $l_1 = l_2$. We will mainly study the effect of BS density $\lambda$ on the resulting temporal interference correlation coefficient $\zeta_I(\ell_1,\ell_1)$, whose expression is given next in Theorem \ref{Thm:Theorem 1}.
\begin{theorem} \label{Thm:Theorem 1}
The temporal interference correlation coefficient ($v = 0$) of the typical user at time slots $1$ and $2$, where the path loss function $g(x)$ is given by (\ref{Eq:Path loss function}) is
\begin{align}
\zeta_I(\ell_1,\ell_1) = \frac{\lambda\int\limits_{\R^2}g^2(x) \ud x - a(x_2)+2\lambda\E\big[g(x_2)\int\limits_{B(0,r_2)} g(x) \ud x\big]}{\E[h^2]\lambda\int\limits_{\R^2}g^2(x) \ud x - b(x_2) +2\lambda\E\big[g(x_2)\int\limits_{B(0,r_2)} g(x) \ud x\big]},
\end{align}
where $a(x_2) = \E[g^2(x_2)]+\E[g(x_2)]^2$, $b(x_2) = \E[h^2]\E[g^2(x_2)]+\E[g(x_2)]^2$,  $r_2 = \|x_2\|$ and $f_{R_2}(r_2) = 2\lambda \pi r_2\exp(-\lambda \pi r_2^2)$.
\end{theorem}
\begin{proof}
For a static user ($v =0$), the serving BS in time slots $1$ and $2$ are the same i.e. $x_1 = x_2$ and is simply the closest BS to the user's location. As a result, there is no handoff to a different serving BS and therefore $\P(H|r_1,\theta) = 0$. The distance distribution of this serving BS in time slots $1$ and $2$ is therefore given by the null probability of PPP $\Phi$ as $f_{R_1}(r_1) = 2 \lambda \pi r_1 e^{-\lambda \pi r_1^2}$ and $f_{R_2}(r_2) = 2\lambda \pi r_2 e^{-\lambda \pi r_2^2}$. From (\ref{T1}),
\begin{align}
T_1 \stackrel{(a)}=  \lambda\E\big[g(x_1)\int_{\R^2 \setminus (\mathcal{C}_1 \cup \mathcal{C}_2)}g(x) \ud x\big] \stackrel{(b)}=  \lambda\E\big[g(x_2)\int_{\R^2 \setminus \mathcal{C}_2}g(x) \ud x\big],
\end{align}
where (a) follows because the second term in (\ref{T1}) goes to zero (no handoff for static user ($v = 0$)) and (b) follows as $x_1 = x_2$ and $\mathcal{C}_1 = \mathcal{C}_2$. Using the same argument in (\ref{T2}), we obtain
\begin{align}
T_2 =  \lambda\E\big[g(x_2)\int_{\R^2 \setminus \mathcal{C}_2}g(x) \ud x\big].
\end{align}
Therefore we obtain $T_1 = T_2$ and substituting their expression in (\ref{I1I2}), we get
\begin{align}
\notag \E[I(1)I(2)] &= \E[I_{\rm t}(1)I_{\rm t}(2)]-\E[g(x_1)g(x_2)]-2\lambda\E\big[g(x_2)\int_{\R^2 \setminus \mathcal{C}_2}g(x) \ud x\big] \\
&= \E[I_{\rm t}(1)I_{\rm t}(2)]-\E[g^2(x_2)]-2\lambda\E\big[g(x_2)\int_{\R^2 \setminus \mathcal{C}_2}g(x) \ud x\big] \label{cross moment}
\end{align}
Substituting the expressions of mean, second moment and first order cross moment of interference from (\ref{Mean Int}), (\ref{I11}) and (\ref{cross moment}) in the definition of correlation coefficient in (\ref{Corr Coeff new}), we obtain the final result.
\end{proof}
For this static user scenario, we now provide asymptotic results on the effect of BS density $\lambda$ on the temporal interference correlation coefficient $\zeta_I(\ell_1,\ell_1)$. 
\begin{corollary} $ \lim\limits_{\lambda \to 0} \: \: \zeta_I(\ell_1,\ell_1) = \frac{1}{\E[h^2]} $ and $ \lim\limits_{\lambda \to \infty} \: \: \zeta_I(\ell_1,\ell_1) = \frac{1}{\E[h^2]}$. \label{Corollary 1}
\begin{proof}
As stated before, the distance of the serving BS for the typical user (static) in both time slots is Rayleigh distributed with its distribution given as $f_{R_2}(r_2) = 2\lambda \pi r_2\exp(-\lambda \pi r_2^2)$. Hence, the serving (closest) BS is located at a mean distance $\E[R_2] = \overbar{R_2} = 0.5/\sqrt{\lambda}$. 

For $\lambda \to \infty$, we have $\overbar{R_2} \to 0$. This asserts that in a  highly dense network, the serving BS is located very close to the typical  user (origin) and hence the integral in the expression of $\zeta_I(\ell_1,\ell_1)$ in Theorem \ref{Thm:Theorem 1} vanishes to zero. Therefore, 
\begin{align}
\lim\limits_{\lambda \to \infty} \zeta_I(\ell_1,\ell_1) = \lim\limits_{\lambda \to \infty} \: \: \frac{\lambda\int\limits_{\R^2}g^2(x) \ud x}{\E[h^2]\lambda\int\limits_{\R^2}g^2(x) \ud x} = \frac{1}{\E[h^2]}.
\end{align}
Similarly for $\lambda \to 0$, we have $\overbar{R_2} \to \infty$ i.e. the closest BS $x_2$ (or $x_1$) is located very far away from the typical user in a sparsely dense network.
\begin{align}
\lim\limits_{\lambda \to 0} \zeta_I(\ell_1,\ell_1) = \lim\limits_{\lambda \to 0} \: \: \frac{\E[g^2(x_1)]+\E[g(x_1)]^2}{\E[h^2]\E[g^2(x_1)]+\E[g(x_1)]^2} \stackrel{(a)} = \frac{1}{\E[h^2]},
\end{align}
where (a) follows because $\E[g(x_1)]^2 \ll \E[g^2(x_1)]$ due to i) Jensen's inequality and ii) the monotonically decreasing behaviour of $g(x)$.
\end{proof}
\end{corollary}
The above result gives insights on the temporal interference correlation in cellular networks under closest BS-association policy for two extreme cases of BS density. The result for the asymptotic cases is the same as an ad-hoc network scenario ($\zeta_I = 1/\E[h^2]$). The spatio-temporal interference correlation coefficient for an ALOHA ad-hoc network is derived in \cite{Ganti} and stated below for ALOHA parameter $p =1$ .
\begin{align}
\zeta_{I}^{\tt ad}(\ell_1,\ell_2) = \frac{\E[{I_{\rm t}(1)}{I_{\rm t}(2)}] - \E[I_{\rm t}(2)]^2}{\E[{I_{\rm t}(2)^2}]-\E[{I_{\rm t}(2)]^2}} = \frac{\int_{\R^2}g(x)g(x-v) \ud x}{\E[h^2]\int_{\R^2}g^2(x) \ud x}
\end{align} 
As noted earlier, we will revisit the general case of $v>0$ as a part of numerical results in Section \ref{Sec:sec5}, where we will  provide further insights by comparing the spatio-temporal interference correlation coefficient for an ad hoc network $\zeta_{I}^{\tt ad}(\ell_1,\ell_2)$ and cellular network $\zeta_{I}(\ell_1,\ell_2)$.
\section{Joint Coverage Probability} \label{sec:sec4}
As studied in Section \ref{sec:sec3}, there is correlation in the interference powers among different spatial locations of the user in a cellular network. Consider a \emph{static user} scenario where a typical user is static at a given spatial location in the network for multiple time slots and connects to its closest BS (serving BS) in each time slot. Due to the temporal correlation of interference, it is seen that if the user is in outage ($1-$coverage) of the serving BS in a given time slot, there is a higher probability that the user will be also be in outage in the future time slots \cite{DiversityPoly}. From this arises the need for correlation-aware retransmission schemes where the BSs do not re-transmit (or remains silent) for certain time slots if an outage is encountered. A suitable metric which measures the correlation in coverage (or outage) in different time slots (or spatial locations) is the \emph{joint coverage probability}, which is defined formally in Definition \ref{Def:def2}.

In cellular networks, as discussed in Section \ref{sec2}, there can either be a handoff to a new serving BS or no handoff as a typical user moves from one location to another. The joint coverage probability of a typical user in the two spatial locations hence depends on the two handoff scenarios. Although handoffs in cellular networks are critical in providing a user with the best serving BS at any given spatial location, excessive handoffs can also result in overheads and handoff delays \cite{Handover}. This is a more pertinent issue in ultra-dense networks \cite{Ultradense}, where a large density of BSs may result in unnecessary handoffs i.e. a handoff to a closer BS even though continued connection to the previous serving BS meet the QoS requirements. Hence, for completeness, we study \emph{handoff skipping} \cite{HandoffSkipping}, where a user skips certain handoffs and remain connected to the same serving BS after moving a certain distance. In this section, we first study the joint coverage probability of a typical user with handoff skipping and then move on to a more \emph{conventional handoff scenario}, where handoffs occur as soon as a user is closer to a new BS. In contrast to prior works which just study joint coverage probability for extreme cases of correlation ($v=0$ and $v\rightarrow \infty$, which respectively correspond to the static and {\em highly} mobile user scenarios), we derive new analytical results for the joint coverage probability for the more relevant case of finite mobility, where $0\leq v < \infty$.
\begin{figure}[t!]
\centering{
\includegraphics[width=.50\linewidth]{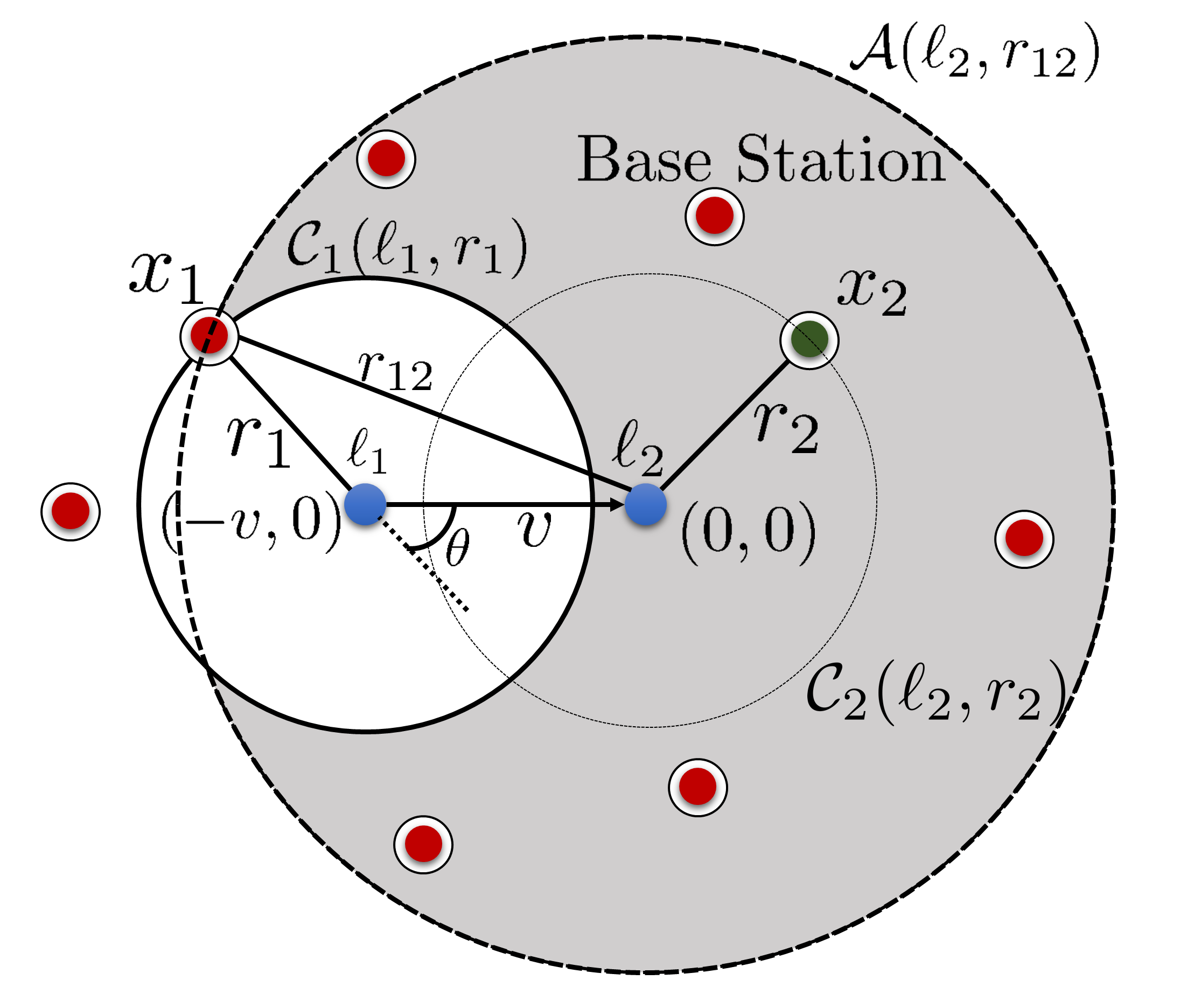}
\caption{System Model when a typical user (denoted by \emph{blue dot}) moves from $\ell_1$ to $\ell_2$ for different handoff strategies. i) When \emph{handoff skipping} is used, the user skips handoffs to closer BSs and remain connected to the BS $x_1$ (serving BS at $\ell_1$) at $\ell_2$. ii) For \emph{conventional handoffs}, the user at $\ell_2$ is handed off to the closest BS $x_2$ (shown in green) at $\ell_2$.}
\label{Fig:HandoffSkipping}
}
\end{figure}
\subsection{Joint Coverage Probability With Handoff Skipping}
Fig. \ref{Fig:HandoffSkipping} depicts handoff skipping scenario where user remains connected to BS $x_1$ after moving from $\ell_1$ to $\ell_2$ even in the presence of other closer BSs (i.e. it skips handoff to those BSs). In case of conventional handoffs, the user would have been handed off to the closest BS $x_2$ at $\ell_2$. In this section, we first derive the joint coverage probability  for the handoff skipping scenario in Theorem \ref{Thm:Theorem 2}.
\begin{theorem} \label{Thm:Theorem 2}
When handoff skipping is used, the joint coverage probability of a typical user at two locations $\ell_1$ and $\ell_2$ separated by a distance $v$ as shown in Fig. \ref{Fig:HandoffSkipping} in a PPP of BS density $\lambda$ is given as:
\begin{align}
\notag \P(\sir_1 > T, \sir_2 > T) &= \E_{R_1,\Theta,\Gamma}\bigg[\exp\bigg(-2\pi\lambda \int_{z_1}^{\infty} F_1(r_1,r,\gamma,\theta) \: r \ud r \bigg ) \\ & \qquad \exp\bigg(2\lambda \int_{|v-r_1|}^{v+r_1} \cos^{-1}\big(\frac{r^2+v^2-r_1^2}{2rv}\big)F_1(r_1,r,\gamma,\theta) \: r\ud r  \bigg)\bigg],
\end{align}
where  $F_1(r_1,r,\gamma,\theta) = 1 - \frac{1}{(1+Tr_1^{\alpha}(r^2+v^2-2rvcos\gamma )^{-\alpha/2})(1+Tr_{12}^{\alpha}r^{-\alpha})}$, $r_{12} = \sqrt{r_1^2+v^2+2r_1v\cos \theta}$ and $z_1 =  {\tt max}(0,r_1-v)$.
\begin{proof}
Appendix \ref{Appendix A}.
\end{proof}
\end{theorem}
Having obtained the expression of joint coverage probability with handoff skipping, we now move on to study the joint coverage under conventional handoffs. More insights on the joint coverage probability will be provided through numerical results in Section \ref{Sec:sec5}.
\subsection{Joint Coverage Probability With Conventional Handoffs}
In this subsection, we derive the joint coverage probability in two spatial locations of a typical user with conventional handoffs. Considering the two scenarios possible when a user moves a distance $v$ at an angle $\theta$ from $\ell_1$ to $\ell_2$ (as shown in Fig. \ref{Fig:SystemModel}), we derive the joint coverage probability for both scenarios individually (no handoff and handoff) to obtain the total joint coverage probability. From total probability theorem, the joint coverage probability ${\rm P_{c}}(\ell_1,\ell_2)$ at two spatial locations $\ell_1$ and $\ell_2$ is given as
\begin{align}
\P(\sir_1 > T, \sir_2 > T) = \P(\sir_1 > T, \sir_2 > T,\bar{H}) + \P(\sir_1 > T, \sir_2 > T,H). \label{Eq:Joint Cov Prob}
\end{align} 
We first derive the expression of the first term i.e. the joint coverage probability under no handoff scenario in Theorem \ref{Thm:Theorem 3} with its proof given in Appendix \ref{Appendix B}.
\begin{theorem} \label{Thm:Theorem 3}
Under a no handoff scenario (Fig. \ref{Fig:SystemModel} (b)), the joint coverage probability of a typical user at the two spatial locations $\ell_1$ and $\ell_2$  in a PPP of BS density $\lambda$  (i.e. same serving BS at $\ell_1$ and $\ell_2$) is given as
\begin{align}
\notag \P(\sir_1 > T, \sir_2 > T,\bar{H}) &= \E_{R_1,\Theta,\Gamma}\bigg[\P(\bar{H}|r_1,\theta)\exp\big(-2\pi\lambda \int_{r_{12}}^{\infty} F_1(r_1,r,\gamma,\theta) \: r \ud r \big ) \\ & \qquad  \exp\big(2\lambda \int_{r_{12}}^{v+r_1} \cos^{-1}\big(\frac{r^2+v^2-r_1^2}{2rv}\big)F_1(r_1,r,\gamma,\theta) \: r\ud r  \big)\bigg], \label{JointCovNoHandoff}
\end{align}
where $F_1(r_1,r,\gamma,\theta)$ and $r_{12}$ are defined in Theorem \ref{Thm:Theorem 2} and $\P(\bar{H}|r_1,\theta) = 1 -\P(H|r_1,\theta)$ is given by (\ref{Eq:Prob handoff}).
\end{theorem}

We now derive the joint coverage probability under the handoff scenario in Theorem \ref{Thm:Theorem 4} with its proof given in Appendix \ref{Appendix C}. 
\begin{theorem} \label{Thm:Theorem 4}
The joint coverage probability of a typical user at two spatial locations $\ell_1$ and $\ell_2$  in a PPP of BS density $\lambda$ under a handoff scenario (Fig. \ref{Fig:SystemModel}(a)) i.e. different serving BS at both locations (BS $x_1$ at $\ell_1$ and BS $x_2$ at $\ell_2$) is given as
\begin{align}
\notag \P(\sir_1 > T, \sir_2 > T,H) &= \E_{R_1,R_2,\Theta,\Gamma}\bigg[\P(H|r_1,\theta)\frac{1}{1+Tr_1^{\alpha}{\|{x_2-v}\|}^{-\alpha}}\frac{1}{1+Tr_2^{\alpha}{\|{x_1}\|}^{-\alpha}} \\ & \: \: \exp\big(-2\pi\lambda \int_{r_{2}}^{\infty} F_2(r_1,r_2,r,\gamma) \: r \ud r \big)\exp\big(\lambda \mathcal{B}_1(r_1,r_2,\gamma)  \big)\bigg],
\end{align}
where $F_2(r_1,r_2,r,\gamma) = 1 - \frac{1}{(1+Tr_1^{\alpha}(r^2+v^2-2rvcos\gamma )^{-\alpha/2})(1+Tr_{2}^{\alpha}r^{-\alpha})}$ and $\mathcal{B}_1(r_1,r_2,\gamma)$ given by (\ref{B1Integral}).
\end{theorem}
Having obtained the joint coverage probability in two user locations $\ell_1$ and $\ell_2$ separated by any distance $v$, we now study the joint coverage for two extreme cases: i) static user ($v =0$) and ii) highly mobile user ($v \rightarrow \infty$). The results are provided below.
\begin{corollary} (Static user)
The joint coverage probability for a static user ${\rm P_{c}}(\ell_1,\ell_1)$ i.e. when a typical user remains at the same spatial location ($\ell_1 = \ell_2$) for 2 different time slots is given as:
\begin{align}
{\rm P_{c}}(\ell_1,\ell_1) = \frac{1}{\,_2F_1\bigg(2,-\frac{2}{\alpha};1-\frac{2}{\alpha};-T\bigg)},
\end{align}
where $\,_2F_1(a,b;c;z)$ is the Gauss hypergeometric function, $\alpha > 2$ is the path loss exponent and $T$ is the $\sir$ threshold.
\begin{proof}
As the user is static ($v = 0$) during both time slots, there is no handoff to a different serving BS at time slot $2$, i.e., the user remains connected to the same BS it was connected in time slot $1$. This can also be verified from Lemma \ref{Lemma 1} that $\P(H|r_1,\theta) = 0$ for $v = 0$. As there is no handoff, the joint coverage probability under handoff scenario is zero i.e. $\P(\sir_1 > T, \sir_2 >T, H) = 0$ (from Theorem \ref{Thm:Theorem 3}). Therefore from (\ref{Eq:Joint Cov Prob}), the joint coverage probability for a static user, ${\rm P_{c}}(\ell_1,\ell_1)$ 
\begin{align}
& = \P(\sir_1 > T, \sir_2 >T, \bar{H}) \\\notag
&\stackrel{(a)}= \E_{R_1,\Theta,\Gamma}\bigg[\exp\big(-2\pi\lambda \int\limits_{r_{1}}^{\infty} F(r_1,r,\gamma,\theta) \: r \ud r \big )\exp\big(2\lambda \int\limits_{r_{1}}^{r_1} \cos^{-1}\big(\frac{r^2+v^2-r_1^2}{2rv}\big)F(r_1,r,\gamma,\theta) \: r\ud r  \big)\bigg] \\\notag
&\stackrel{(b)}= \E_{R_1}\bigg[\exp\big(-2\pi\lambda \int\limits_{r_{1}}^{\infty} \left(1 - \frac{1}{(1+Tr_1^{\alpha}r^{-\alpha})^2}\right) r\ud r\bigg]
\end{align}
where (a) follows from (\ref{JointCovNoHandoff}) and using $r_{12} = r_1$ and $P(\bar{H}|r_1,\theta) = 1$, and (b) results by using the definition of $F(r_1,r,\gamma,\theta)$ from Theorem \ref{Thm:Theorem 2} and substituting $v = 0$. The final result follows by deconditioning w.r.t. $r_1$ using $f_{R_1}(r_1) = 2\lambda \pi r_1 e^{-\lambda \pi r_1^2}$, some algebraic manipulations and the definition of Gauss hypergeometric function.
\end{proof}
\end{corollary}
\begin{corollary} (Highly mobile user)
The joint coverage probability for a highly mobile user ($v \to \infty$) i.e. when user moves a large distance between $\ell_1$ and $\ell_2$ is given as
\begin{align}
\lim_{v \to \infty} {\rm P_{c}}(\ell_1,\ell_2) = \bigg(\frac{1}{1+\rho(T,\alpha)}\bigg)^2,
\end{align}
where $\rho(T,\alpha) = T^{2/\alpha} \int_{T^{-2/\alpha}}^{\infty} \frac{\ud u}{1+u^{\alpha/2}} $.
\begin{proof}
For a highly mobile user ($v \to \infty$), there is always handoff as the user moves a large distance from $\ell_1$ to $\ell_2$ i.e. $P(H|r_1,\theta) = 1$ from (\ref{Eq:Prob handoff}). Also, $F_{R_2}(r_2) = 1- e^{-\lambda \pi r_2^2}$ from Lemma \ref{Lemma 2} as $|\mathbf{C}(\ell_1,r_1,\ell_2,r_{2},v)| = \pi r_2^2$ using Definition \ref{Definition 1} ($v \to \infty$ correspond to disjoint circle case as per Remark \ref{Remark 1}). Using the above expressions in (\ref{Eq:InfiniteMobility}), we obtain $\lim_{v \to \infty} \P(\sir_1 > T, \sir_2 > T,\bar{H}) $
\begin{align}
\notag &= \lim_{v \to \infty} \E_{R_1,R_2}\bigg[ \exp\big(-T{r_1}^{\alpha} \sum_{x \in \Phi \backslash \{x_1\}}h_{x}(1) {\|{x-v}\|}^{-\alpha}\big) \exp\big(-T{r_{2}}^{\alpha} \sum_{x \in \Phi \backslash \{x_2\}}h_{x}(2) {\|{x}\|}^{-\alpha}\big)\bigg] \\\notag
& \stackrel{(a)}= \E_{R_1}\bigg[ \exp\big(-T{r_1}^{\alpha} \sum_{y \in \Phi' \backslash \{x_1\}}h_{y}(1) {\|{y}\|}^{-\alpha}\big)\bigg] \E_{R_2} \bigg[ \exp\big(-T{r_{2}}^{\alpha} \sum_{x \in \Phi \backslash \{x_2\}}h_{x}(2) {\|{x}\|}^{-\alpha}\big)\bigg] \\\notag
& \stackrel{(b)}= \E_{R_1}\bigg[\exp\bigg( -\lambda \int\limits_{r_1}^{\infty}  \bigg(1 - \frac{1}{1+Tr_{1}^{\alpha}{u}^{-\alpha}}\bigg) u \ud u \bigg)\bigg] \E_{R_2}\bigg[\exp\bigg( -\lambda \int\limits_{r_2}^{\infty}  \bigg(1 - \frac{1}{1+Tr_{2}^{\alpha}{v}^{-\alpha}}\bigg) v \ud v \bigg)\bigg] 
\end{align}
where (a) follows from the fact that under $v \rightarrow \infty$, two different instances of the point process are observed at the two locations, which allows us to distribute the expectation across the two terms. For notational simplicity, we denote the {\em translated} PPP as $\Phi' = \{x-v\}$.  Step (b) follows as  $h_y(1) \sim \exp(1)$, $h_x(2) \sim \exp(1)$ and using the PGFL of PPP $\Phi$ and $\Phi'$. The final result follows by deconditioning w.r.t. $R_1$ and $R_2$ after some change of variables and algebraic manipulations.
\end{proof}
\end{corollary}
\begin{figure}[t!]
\centering
\begin{subfigure}{}
\centering
\includegraphics[width=0.47 \linewidth]{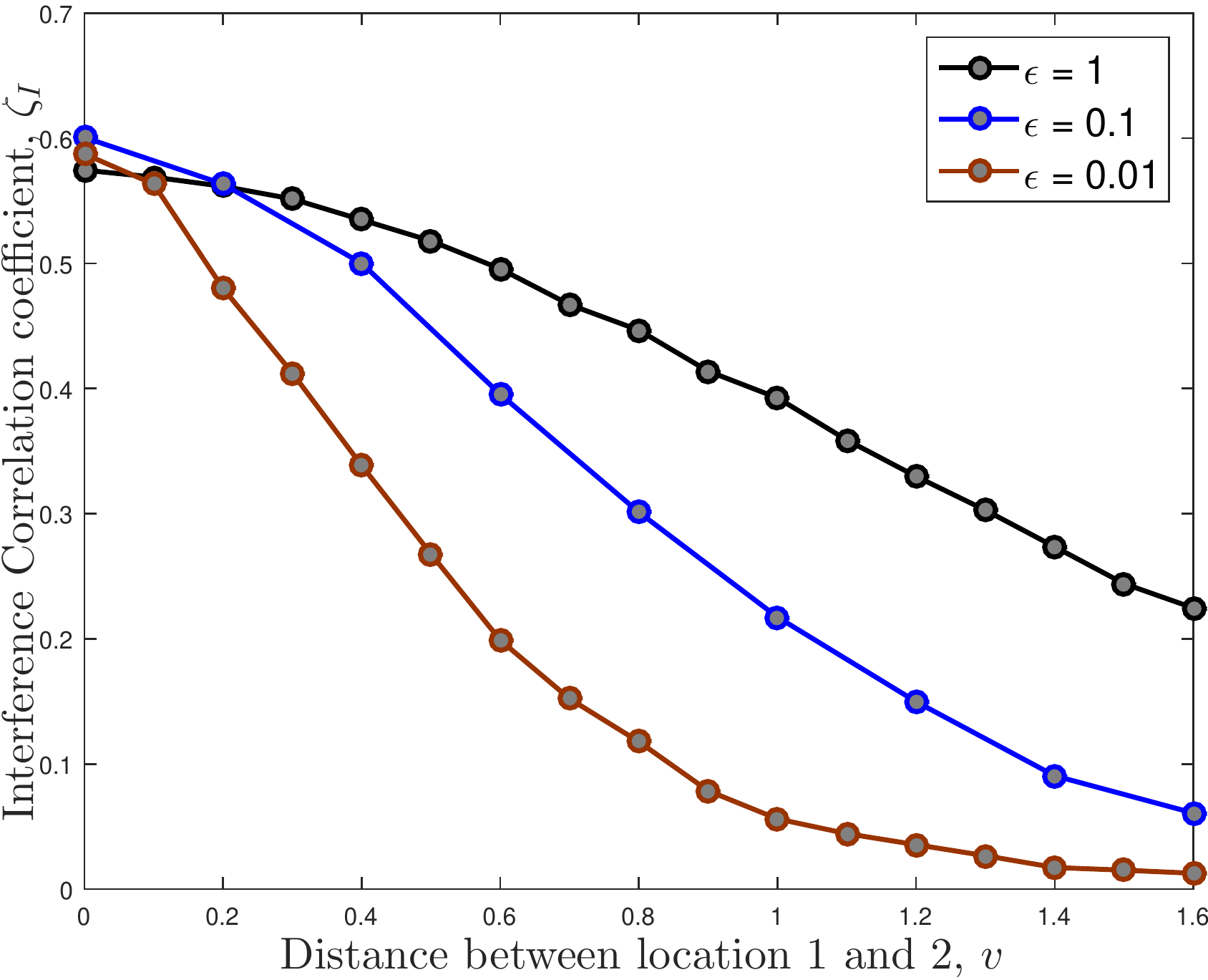}
\label{Fig:CorrCoeffCellular}
\end{subfigure}
\begin{subfigure}{}
\centering
\includegraphics[width=0.47 \linewidth]{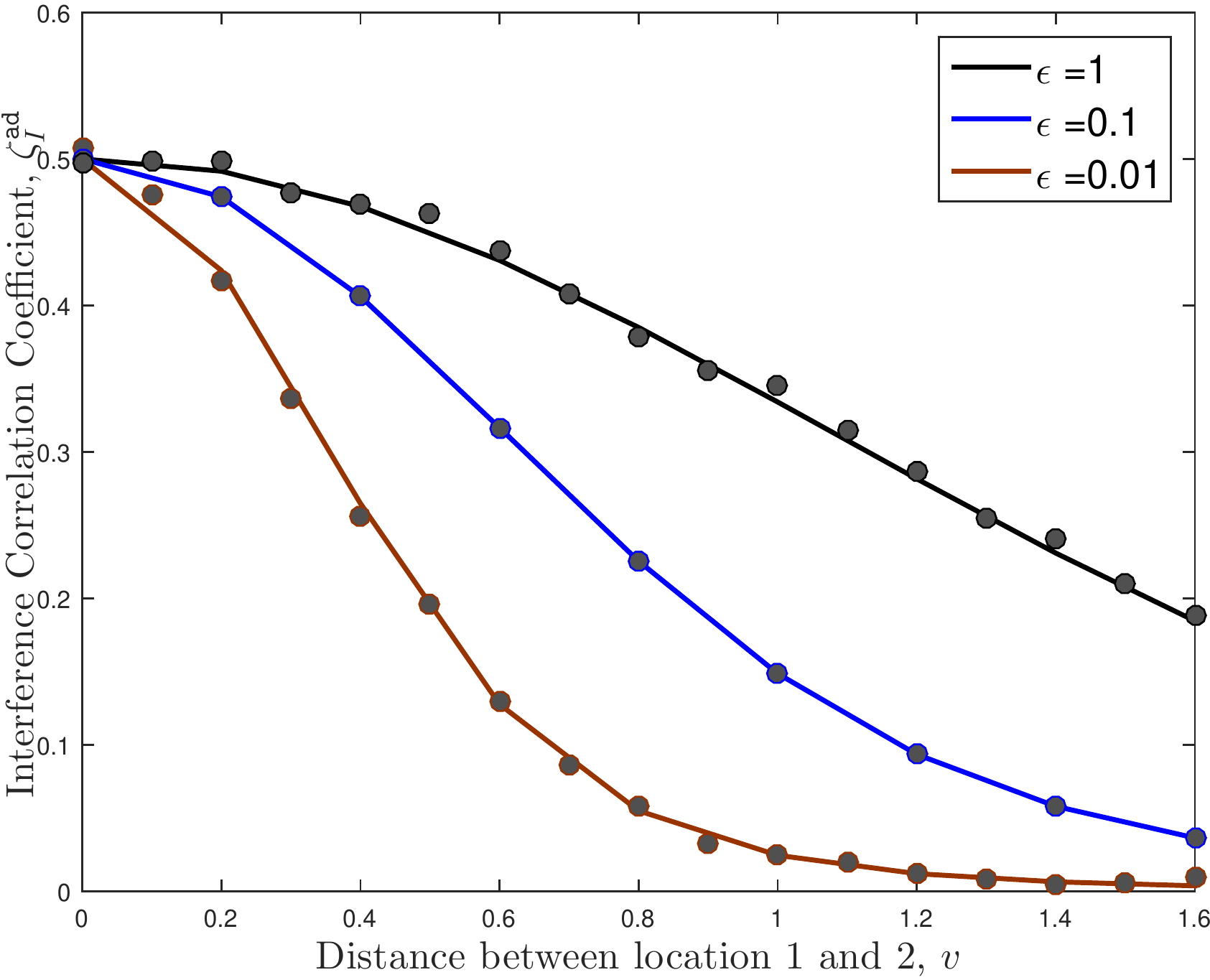}
\label{Fig:Asym}
\end{subfigure}
\caption{ Effect of $v$ on interference correlation  coefficient in (\emph{left}) cellular networks and (\emph{right}) ad hoc networks. Here, $\epsilon$ is the path-loss function parameter.}
\label{Fig:CorrCoeff}
\end{figure}

\section{Results and Discussion} \label{Sec:sec5}
In this section, we validate the accuracy of the analytical results (interference correlation coefficient and joint coverage probability) by means of simulations. In all simulations, we set the SIR threshold, $T$ as 0 dB and path loss exponent $\alpha = 4$, unless mentioned otherwise.
\subsection{Spatio-temporal interference correlation coefficient}
\subsubsection{Effect of distance $v$} Fig. \ref{Fig:CorrCoeff} (\emph{left}) and (\emph{right}) plot the interference correlation coefficient between two spatial locations $\ell_1$ and $\ell_2$ separated by a distance $v$ for cellular networks and ad hoc networks respectively. The interference correlation coefficient decreases with the distance between the two spatial locations. This coincides with our intuition that the set of interferers for two closeby user locations are similar resulting in a higher correlation, while independent interferers for spatial locations far apart result in lower interference correlation. We observe that the correlation coefficient attains the maximum vale for $v = 0$ (same spatial location or static user) which corresponds to the temporal correlation coefficient as derived in Theorem \ref{Thm:Theorem 1}. For large $v$, the correlation coefficient approaches to zero signifying uncorrelated interference powers for far away spatial locations. 
\subsubsection{Effect of path loss function parameter $\epsilon$}
As evident from Fig. \ref{Fig:CorrCoeff} (\emph{right}), interference correlation in ad hoc networks decreases with higher path loss i.e. lower path loss function parameter $\epsilon$. With higher path loss, the interference is dominated more by the transmitters closer to the user and therefore the correlation among interferers decreases overall. However, as seen from Fig. \ref{Fig:CorrCoeff} (\emph{left}), interference correlation in cellular networks does not exhibit an even trend with $\epsilon$. This is because the interference in cellular networks depends on the choice of the serving BS (closest BS) at any given spatial location as well as the path loss function. For small $v$, there is a higher probability of connecting to the same BS at both locations and thereby is a major factor in deciding interference correlation at the two spatial locations. Hence there is no such trend of interference correlation coefficient with $\epsilon$ for small $v$. However for large $v$, the two locations $\ell_1$ and $\ell_2$ are far apart (different serving BSs), which means the interference correlation depends primarily on the path loss function and decreases with $\epsilon$ like in ad hoc networks.
\subsubsection{Effect of BS density $\lambda$}
Fig. \ref{Fig:EffectOflambda} plots the effect of BS density $\lambda$ on the temporal interference correlation coefficient ($v =0$) in cellular networks. It can be seen from the figure that the correlation coefficient exhibits a bell-curve trend w.r.t. BS density $\lambda$ i.e. interference correlation increases with BS density, attains a peak and then decreases with further increase in BS density. This behaviour is not observed for ad hoc networks, where the temporal interference correlation is independent of node density. However in cellular networks, this bell curve trend signifies a non-intuitive result that there is a certain BS density $\lambda^*$ for which interference correlation is maximized and this density $\lambda^*$ varies w.r.t the path loss parameter $\epsilon$. As $\epsilon$ increases, the large-scale path loss $g(x)$ decreases, thereby requiring a lower BS density $\lambda^*$ to attain a high interference correlation. For $\epsilon \to 0$ (singular path-loss function), it requires an extremely large BS density $\lambda^*$ for maximum interference correlation and thereby interference correlation does not change w.r.t. $\lambda$ and remains at $1/\E[h^2] = 0.5$ for a Rayleigh fading channel.
\begin{figure}[t!]
\centering{
\includegraphics[width=.60\linewidth]{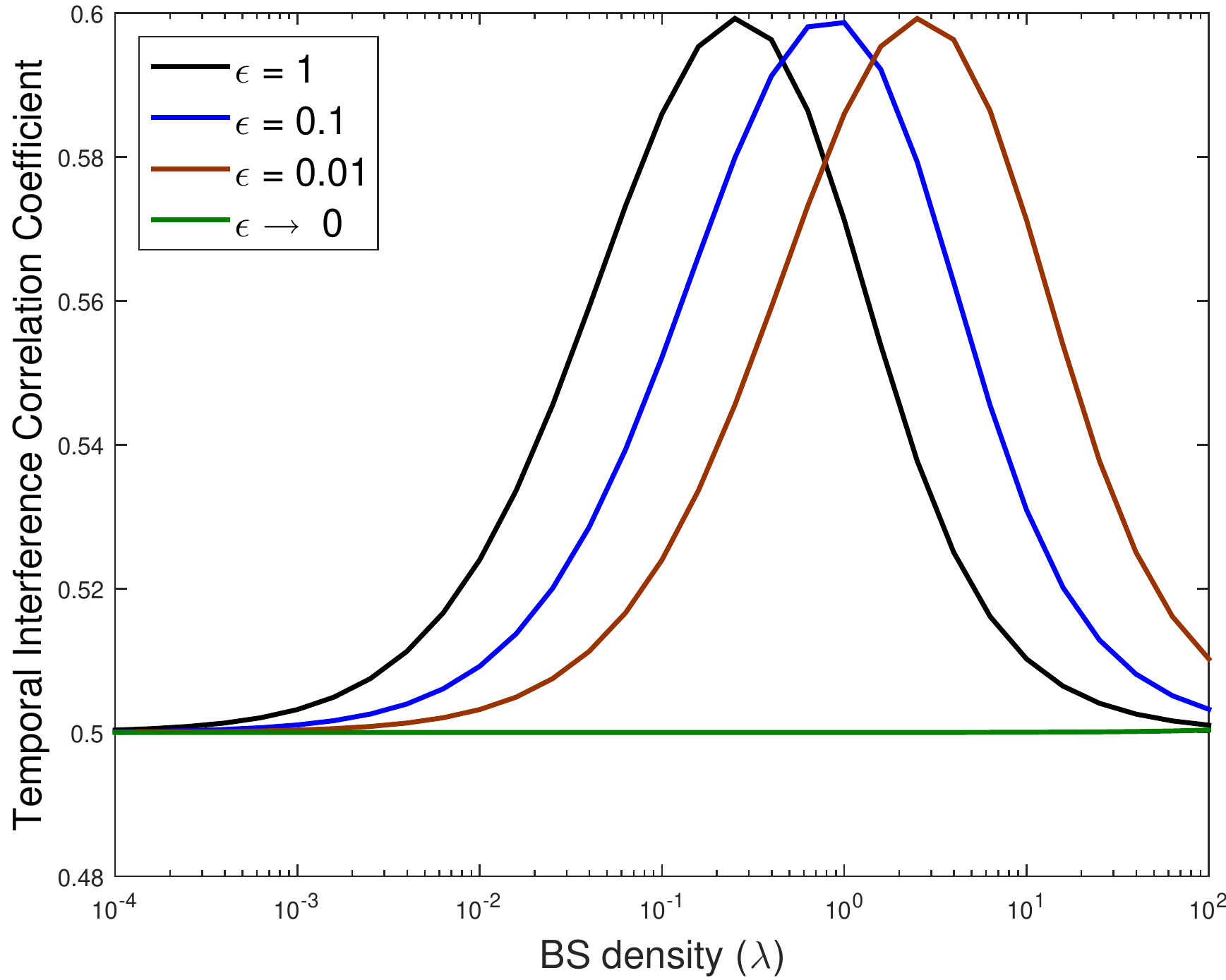}
\caption{Effect of BS density $\lambda$ on temporal interference correlation  coefficient ($v =0$) for a Rayleigh fading channel ($\E[h^2] = 2$).}
\label{Fig:EffectOflambda}
}
\end{figure}
\subsection{Joint coverage probability}
Fig. \ref{Fig:JCP} plots the joint coverage probability of a typical user with conventional handoffs at two spatial locations $\ell_1$ and $\ell_2$ seperated by a distance $v$. As evident from the figure, the joint coverage probability decreases with distance $v$ between the two spatial locations. This is explained by the decrease in interference correlation with distance $v$ as was seen in Fig. \ref{Fig:CorrCoeff}. With higher correlation, there is a higher chance of being in coverage at the second location $\ell_2$ given that the user is in coverage at the first spatial location $\ell_1$. However in an uncorrelated scenario (far away spatial locations, i.e., $v \rightarrow \infty$), the coverage probability at $\ell_2$ is independent of the coverage at $\ell_1$, which means the joint coverage probability is simply the product of individual coverage probabilities. This trend is evident from Fig. \ref{Fig:JCP}, where the joint coverage probability at $\ell_1$ and $\ell_2$ decreases from a completely correlated scenario ($v = 0$) and approaches an uncorrelated scenario  for large $v$.

Fig. \ref{Fig:JCPHandoff} compares the joint coverage probability of a typical user with handoff skipping and conventional handoffs. When handoff skipping is used, the joint coverage decreases rapidly with $v$ compared to a conventional handoff scenario. This is because of the increase in the number of interfering BSs located closer to the user than the farther located serving BS (due to handoff skipping, user connects to the same serving BS even after displacement).
Although joint coverage probability decreases rapidly when handoffs are skipped, we can avoid handoffs till a certain user displacement if the QoS is tolerable, which will naturally reduce excessive handoff delays. 

\begin{figure}[t!]
\centering{
\includegraphics[width=.50\linewidth]{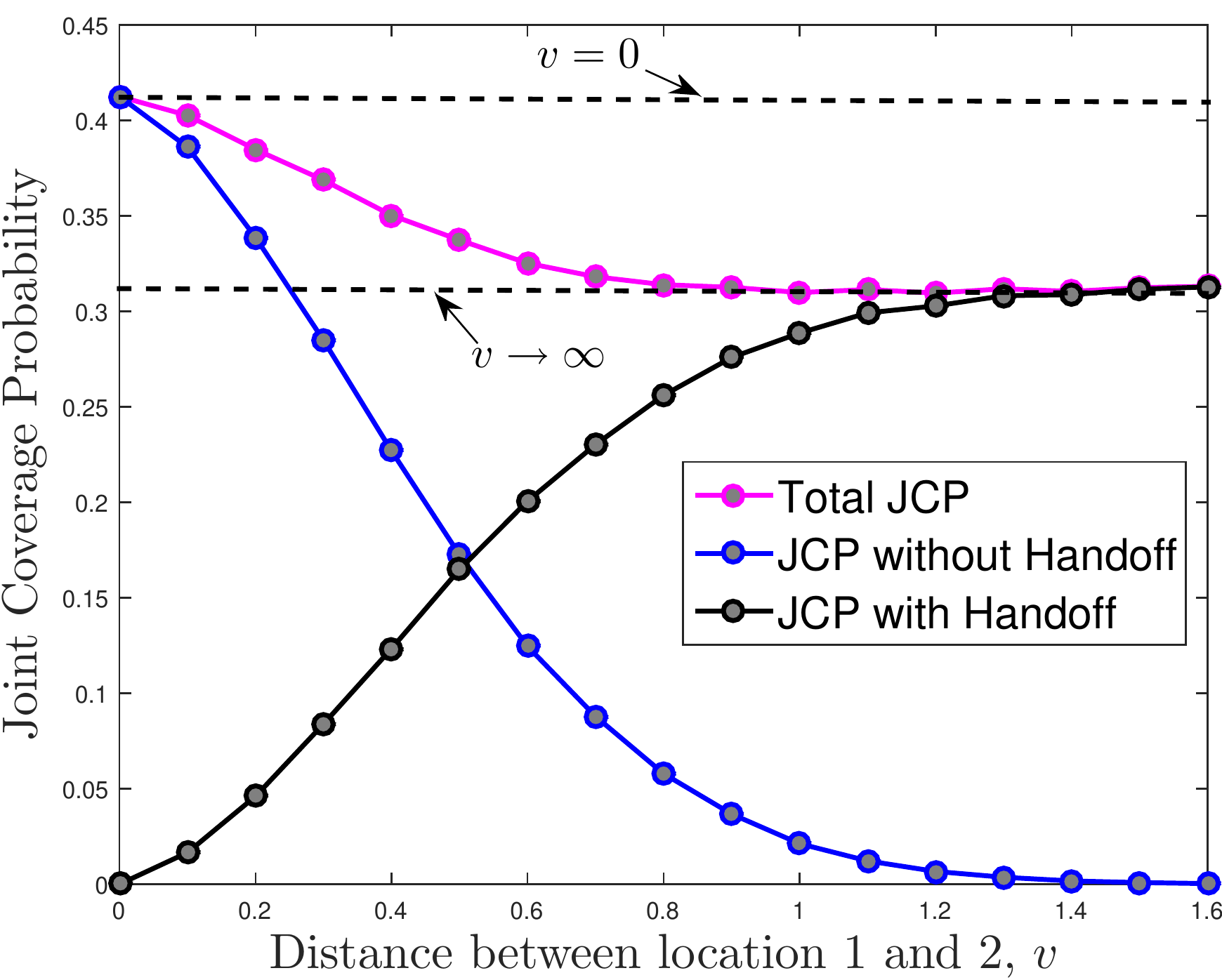}
\caption{Effect of $v$ on joint coverage probability (conventional handoff) in two spatial locations $\ell_1$ and $\ell_2$ ($\lambda = 1 ,\epsilon = 0, T = 0 \: {\rm dB}$).}
\label{Fig:JCP}
}
\end{figure}
\begin{figure}[t!]
\centering{
\includegraphics[width=.50\linewidth]{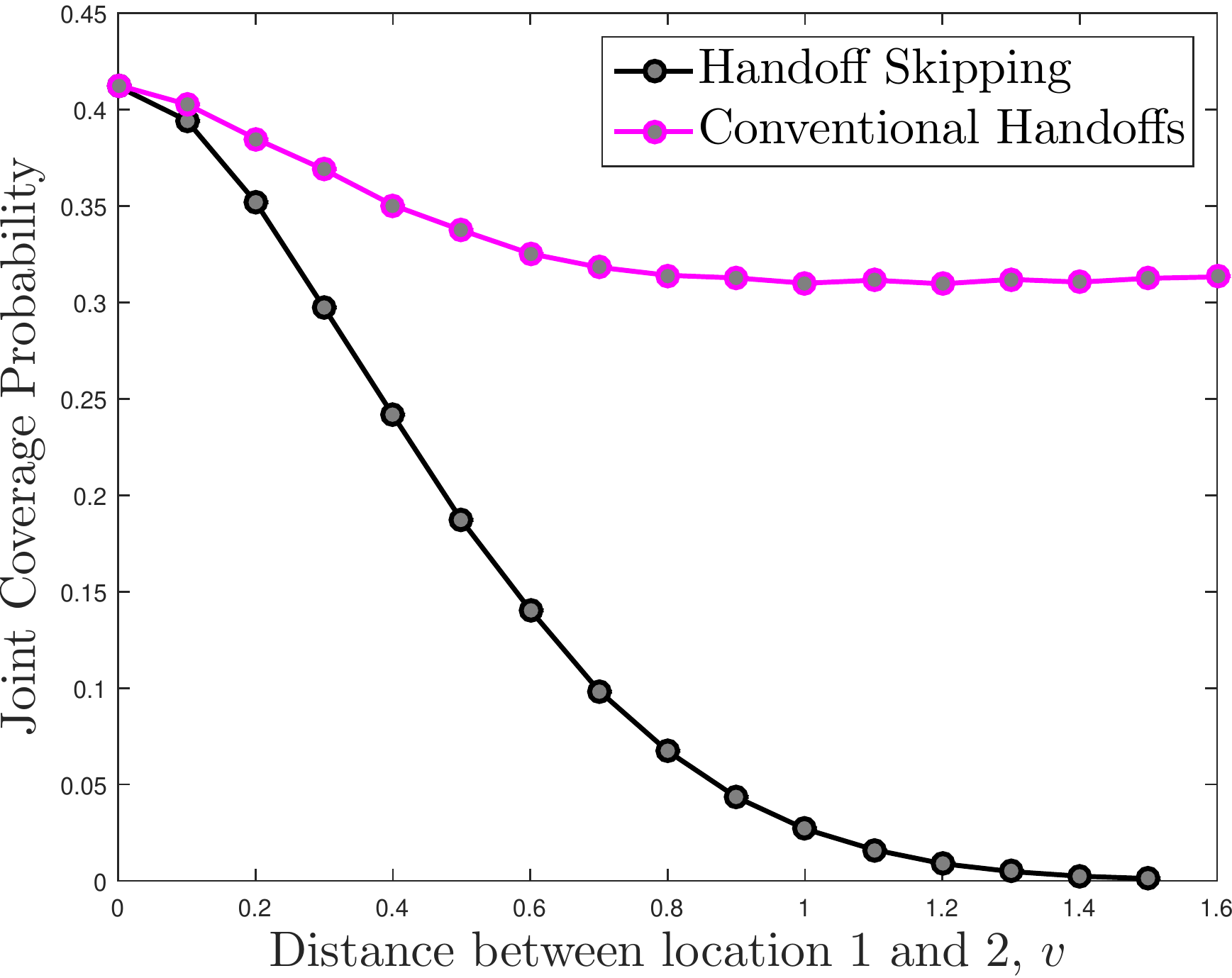}
\caption{Effect of $v$ on joint coverage probability with handoff skipping and conventional handoffs ($\lambda = 1 ,\epsilon = 0, T = 0 \: {\rm dB}$).}
\label{Fig:JCPHandoff}
}
\end{figure}
\section{Concluding Remarks}
In this paper, we provided first comprehensive framework to study spatio-temporal correlation in the interference power as well as the joint coverage probability as observed at two spatial locations in a cellular network. Considering {\em closest BS association} policy for the user at both the locations, we first characterized distributions of the distances from the two locations to their respective serving BSs. Using these results, we then derived expressions for the mean, second moment and first order cross moment of interference that ultimately led to the derivation of spatio-temporal interference correlation coefficient. As expected, interference correlation was shown to decrease with increasing distance between the two locations, eventually approaching zero when the distance between the two locations approached infinity. In order to study correlation in link successes at two spatial locations at a finite distance apart, we then derived exact results for the joint coverage probability at these two user locations under {\em handoff skipping} and {\em conventional handoff} strategies. As expected, joint coverage decreased with distance (less correlation) becoming independent at very far distances. As evident from the analysis, the analytical framework required to characterize these correlations in cellular networks is significantly more complex than its ad hoc network counterpart due to the need to carefully handle cell association policies that complicates the characterization of interference field as observed from the two spatial locations.


This work has many possible extensions. First and foremost, although this work provides exact analysis of the joint coverage probability and spatio-temporal interference correlation coefficient, the resulting expressions are not in closed form and require numerical integrations as is usually the case in most of the stochastic geometry-based analyses. While it is important to perform exact analyses for mathematical completeness as well as to complement (in some cases, circumvent) the system level simulations, it is equally important to extend such works and derive easy-to-use approximations and bounds that enable the readers to draw even better insights. From the system model side, we assumed two arbitrary spatial locations that were a distance $v$ apart. One possible extension could be to endow this separation with a distribution and study how interference correlation varies across multiple spatial locations thereby modeling an actual mobile user.


\appendix
\subsection{Proof of Theorem \ref{Thm:Theorem 2}} \label{Appendix A}
In a network where handoff skipping is used, the joint coverage probability of a typical user which moves a distance $v$ from $\ell_1$ to $\ell_2$ is given as :
\small
\begin{align*} 
&\P(\sir_1 > T, \sir_2 > T) \\
&= \E_{R_1,\Theta}[\P(\sir_1 > T, \sir_2 > T|r_1,\theta)] \\
& \stackrel{(a)}= \E_{R_1,\Theta}[P (h_{x_1}(1){r_1}^{-\alpha} > T  I(1),h_{x_1}(2){r_{12}}^{-\alpha} > T  I(2)|r_1,\theta)]\\
&\stackrel{(b)}= \E_{R_1,\Theta}\bigg[ \exp\big(-T{r_1}^{\alpha} \sum_{x \in \Phi \backslash \{x_1\}}h_{x}(1) {\|{x-v}\|}^{-\alpha}\big) \exp\big(-T{r_{12}}^{\alpha} \sum_{x \in \Phi \backslash \{x_1\}}h_{x}(2) {\|{x}\|}^{-\alpha}\big)\bigg]\\
&= \E_{R_1,\Theta}\bigg[\prod_{x \in \Phi \backslash \{x_1\}} \exp\big( -Tr_1^{\alpha}h_x(1){\|{x-v}\|}^{-\alpha}) \exp(-Tr_{12}^{\alpha}h_x(2){\|{x}\|}^{-\alpha}\big)\bigg] \\\notag
& \stackrel{(c)}= \E_{R_1,\Theta}\bigg[\exp\bigg( -\lambda \int\limits_{\R^2 \backslash \mathcal{C}_1}  1 - \frac{1}{(1+Tr_1^{\alpha}{\|{x-v}\|}^{-\alpha})(1+Tr_{12}^{\alpha}{\|{x}\|}^{-\alpha})}\bigg)\bigg] \\\notag
& \stackrel{(d)}= \E_{R_1,\Theta,\Gamma}\bigg[\exp\bigg( -\lambda \int\limits_{\R^2 \backslash (\mathcal{C}_3 \cup (\mathcal{C}_1 \setminus \mathcal{C}_3))}  \underbrace{\left(1 - \frac{1}{(1+Tr_1^{\alpha}(r^2+v^2-2rvcos\gamma )^{-\alpha/2})(1+Tr_{12}^{\alpha}r^{-\alpha})}\right)}_{F_1(r_1,r,\gamma,\theta)} r \ud r \bigg)\bigg] \\\notag
&\stackrel{(e)}= \E_{R_1,\Theta,\Gamma}\bigg[\exp\bigg(-\lambda \int\limits_{\R^2 \backslash C_3} F_1(r_1,r,\gamma,\theta) r \ud r\bigg)\exp\bigg(\lambda \int\limits_{C_1 \setminus C_3} F_1(r_1,r,\gamma,\theta)r \ud r\bigg)\bigg]  \\\notag
&\stackrel{(f)}= \E_{R_1,\Theta,\Gamma}\bigg[\exp\bigg(-2\pi\lambda \int\limits_{z_1}^{\infty} F_1(r_1,r,\gamma,\theta) \: r \ud r \bigg )\exp\bigg(2\lambda \int\limits_{|v-r_1|}^{v+r_1} \cos^{-1}\big(\frac{r^2+v^2-r_1^2}{2rv}\big)F_1(r_1,r,\gamma,\theta) \: r\ud r  \bigg)\bigg]
\end{align*}
\normalsize
where (a) follows from the definition of $\sir_1$ and $\sir_2$ in (\ref{Eq:SIR1}) and (\ref{Eq:SIR2}) respectively and considering that the serving BS at $\ell_2$ is at a distance $r_{12}$ (same serving BS $x_1$ at $\ell_1$ and $\ell_2$ due to handoff skipping). Step (b) follows from the definition of $I(1)$ and $I(2)$ in (\ref{Eq:I1}) and (\ref{Eq:I2}) respectively and the spatial independence of the fading links, (c) follows from $h_x(1) \sim \exp(1)$, $h_x(2) \sim \exp(1)$ and observing that the interference from the PPP $\Phi$ except $x_1$ is equivalent to considering an exclusion zone $\mathcal{C}_1$ in the network (as no BSs lie inside $\mathcal{C}_1$). Step (d) follows by converting the integral from Cartesian to polar coordinates where $\Gamma$ is a uniform RV in $[0,\pi]$ and denote the angle of interferer $x \in \Phi$ w.r.t. user at $\ell_2$. We also express $\mathcal{C}_1$ as the union of $\mathcal{C}_3$ and $\mathcal{C}_1 \setminus \mathcal{C}_3$ in (d), where $\mathcal{C}_3 = B(0,z_1)$ and $z_1 = \max(0,r_1-v)$ (See Fig. \ref{Fig:Scenarios}). We split the integral into the two regions in (e), while the final result follows by using the law of cosines and appropriate limits of integration for the two regions.
\subsection{Proof of Theorem \ref{Thm:Theorem 3}} \label{Appendix B}
The joint coverage probability of a typical user under the no handoff scenario (event $\bar{H}$) is given as
\small
\begin{align}
\notag &\P(\sir_1 > T, \sir_2 > T,\bar{H}) \\\notag
&= \E_{R_1,\Theta}[\P(\sir_1 > T, \sir_2 > T,\bar{H}|r_1,\theta)] \\\notag
&= \E_{R_1,\Theta}[P (h_{x_1}(1){r_1}^{-\alpha} > T  I(1),h_{x_1}(2){r_{2}}^{-\alpha} > T  I(2)|\bar{H},r_1,\theta)\P(\bar{H}|r_1,\theta)]\\\notag
&= \E_{R_1,\Theta}\bigg[ \exp\big(-T{r_1}^{\alpha} \sum_{x \in \Phi \backslash \{x_1\}}h_{x}(1) {\|{x-v}\|}^{-\alpha}\big) \exp\big(-T{r_{12}}^{\alpha} \sum_{x \in \Phi \backslash \{x_1\}}h_{x}(2) {\|{x}\|}^{-\alpha}\big)\P(\bar{H}|r_1,\theta)\bigg]\\\notag
&= \E_{R_1,\Theta}\bigg[\P(\bar{H}|r_1,\theta)\prod_{x \in \Phi \backslash \{x_1\}} \exp\big( -Tr_1^{\alpha}h_x(1){\|{x-v}\|}^{-\alpha}) \exp(-Tr_{12}^{\alpha}h_x(2){\|{x}\|}^{-\alpha}\big)\bigg] \\\notag
& \stackrel{(a)}= \E_{R_1,\Theta}\bigg[\P(\bar{H}|r_1,\theta)\exp\bigg( -\lambda \int\limits_{\R^2 \backslash \mathcal{C}_1\cup \mathcal{C}_2}  1 - \frac{1}{(1+Tr_1^{\alpha}{\|{x-v}\|}^{-\alpha})(1+Tr_{12}^{\alpha}{\|{x}\|}^{-\alpha})}\bigg)\bigg] \\ \label{Eq:F defn}
& = \E_{R_1,\Theta,\Gamma}\bigg[\P(\bar{H}|r_1,\theta)\exp\bigg( -\lambda \int\limits_{\R^2 \backslash \mathcal{C}_1\cup \mathcal{C}_2}  \underbrace{\left(1 - \frac{1}{(1+Tr_1^{\alpha}(r^2+v^2-2rvcos\gamma )^{-\alpha/2})(1+Tr_{12}^{\alpha}r^{-\alpha})}\right)}_{F_1(r_1,r,\gamma,\theta)} r \ud r \bigg)\bigg] \\\notag
&\stackrel{(b)}= \E_{R_1,\Theta,\Gamma}\bigg[\P(\bar{H}|r_1,\theta)\exp\bigg(-\lambda \int\limits_{\R^2 \backslash C_2} F_1(r_1,r,\gamma,\theta) r \ud r\bigg)\exp\bigg(\lambda \int\limits_{C_2 \setminus C_1} F_1(r_1,r,\gamma,\theta)r \ud r\bigg)\bigg]  \\\notag
&\stackrel{(c)}= \E_{R_1,\Theta,\Gamma}\bigg[\P(\bar{H}|r_1,\theta)\exp\big(-2\pi\lambda \int\limits_{r_{12}}^{\infty} F_1(r_1,r,\gamma,\theta) \: r \ud r \big )\exp\big(2\lambda \int\limits_{r_{12}}^{v+r_1} \cos^{-1}\big(\frac{r^2+v^2-r_1^2}{2rv}\big)F_1(r_1,r,\gamma,\theta) \: r\ud r  \big)\bigg] 
\end{align}
\normalsize
where (a) follows by observing that the interference from PPP $\Phi$ except the serving BS $x_1$ in a no handoff scenario (see Fig. \ref{Fig:SystemModel} (b)) is equivalent to considering an exclusion zone $\mathcal{C}_1 \cup \mathcal{C}_2$  in the network (no BSs lie inside $\mathcal{C}_1 \cup \mathcal{C}_2$). Step (b) follows by splitting the integral into two integration regions, while the final step (c) follows by using appropriate limits of integration for the two regions and using the law of cosines in the second integral.
\subsection{Proof of Theorem \ref{Thm:Theorem 4}} \label{Appendix C}
By definition, the joint coverage probability of a typical user under handoff scenario (event $H$) is given as:
\small
\begin{align}
\notag &\P(\sir_1 > T, \sir_2 > T,H) \\\notag
&= \E_{R_1,\Theta}[\P(\sir_1 > T, \sir_2 > T,H|r_1,\theta)] \\\notag
& \stackrel{(a)}= \E_{R_1,\Theta}[P (h_{x_1}(1){r_1}^{-\alpha} > T  I(1),h_{x_2}(2){r_{2}}^{-\alpha} > T  I(2)|H,r_1,\theta)\P(H|r_1,\theta)]\\  \label{Eq:InfiniteMobility}
&= \E_{R_1,R_2|H,\Theta}\bigg[ \exp\big(-T{r_1}^{\alpha} \sum_{x \in \Phi \backslash \{x_1\}}h_{x}(1) {\|{x-v}\|}^{-\alpha}\big) \exp\big(-T{r_{2}}^{\alpha} \sum_{x \in \Phi \backslash \{x_2\}}h_{x}(2) {\|{x}\|}^{-\alpha}\big)\P(H|r_1,\theta)\bigg] \\\notag
& \stackrel{(b)}= \E_{R_1,R_2|H,\Theta}\bigg[\P(H|r_1,\theta) \exp(-T{r_1}^{\alpha} h_{x_2}(1) {\|{x_2-v}\|}^{-\alpha}) \exp(-T{r_2}^{\alpha} h_{x_1}(2) {\|{x_1}\|}^{-\alpha}) \\\notag & \qquad \qquad \prod_{x \in \Phi \backslash (\{x_1\} \cup \{x_2\}) } \exp\big( -Tr_1^{\alpha}h_x(1){\|{x-v}\|}^{-\alpha}) \exp(-Tr_{2}^{\alpha}h_x(2){\|{x}\|}^{-\alpha}\big)\bigg] \\\notag
& \stackrel{(c)}= \E_{R_1,R_2|H,\Theta}\bigg[\P(H|r_1,\theta)\frac{1}{1+Tr_1^{\alpha}{\|{x_2-v}\|}^{-\alpha}}\frac{1}{1+Tr_2^{\alpha}{\|{x_1}\|}^{-\alpha}} \\\notag & \qquad \qquad \exp\bigg( -\lambda \int\limits_{\R^2 \backslash \mathcal{C}_1\cup \mathcal{C}_2}  1 - \frac{1}{(1+Tr_1^{\alpha}{\|{x-v}\|}^{-\alpha})(1+Tr_{2}^{\alpha}{\|{x}\|}^{-\alpha})}\bigg)\bigg] \\\notag
& = \E_{R_1,R_2|H,\Theta,\Gamma}\bigg[\P(H|r_1,\theta)\frac{1}{1+Tr_1^{\alpha}{\|{x_2-v}\|}^{-\alpha}}\frac{1}{1+Tr_2^{\alpha}{\|{x_1}\|}^{-\alpha}} \\\notag & \qquad \qquad  \exp\bigg( -\lambda \int\limits_{\R^2 \backslash \mathcal{C}_1\cup \mathcal{C}_2}  \underbrace{\left(1 - \frac{1}{(1+Tr_1^{\alpha}(r^2+v^2-2rvcos\gamma )^{-\alpha/2})(1+Tr_{2}^{\alpha}r^{-\alpha})}\right)}_{F_2(r_1,r_2,r,\gamma)} r \ud r \bigg)\bigg] \\\notag
&= \E_{R_1,R_2|H,\Theta,\Gamma}\bigg[\P(H|r_1,\theta)\frac{1}{1+Tr_1^{\alpha}{\|{x_2-v}\|}^{-\alpha}}\frac{1}{1+Tr_2^{\alpha}{\|{x_1}\|}^{-\alpha}} \\\notag & \qquad \qquad \exp\bigg(-\lambda \int_{\R^2 \backslash C_2} F_2(r_1,r_2,r,\gamma) r \ud r\bigg)\exp\bigg(\lambda \underbrace{\int_{C_1 \setminus C_2} F_2(r_1,r_2,r,\gamma)r \ud r}_{\mathcal{B}_1(r_1,r_2,\gamma)} \bigg)\bigg]  \\\notag
& \stackrel{(d)}= \E_{R_1,R_2|H,\Theta,\Gamma}\bigg[\P(H|r_1,\theta)\frac{1}{1+Tr_1^{\alpha}{\|{x_2-v}\|}^{-\alpha}}\frac{1}{1+Tr_2^{\alpha}{\|{x_1}\|}^{-\alpha}} \\\notag & \qquad \qquad \exp\bigg(-2\pi\lambda \int\limits_{r_{2}}^{\infty} F_2(r_1,r_2,r,\gamma) \: r \ud r \bigg )\exp\bigg(\lambda \mathcal{B}_1(r_1,r_2,\gamma)  \bigg)\bigg]
\end{align}
\normalsize
where (a) follows by using the definition of $\sir_1$ and $\sir_2$ for a handoff scenario (serving BS $x_1$ at distance $r_1$ at $\ell_1$ and serving BS $x_2$ at distance $r_2$ at $\ell_2$) and conditioning w.r.t. event $H$ (occurence of handoff). Step (b) follows by splitting the interference into three parts: i) First term corresponds to interference from BS $x_2$ at $\ell_1$, ii) Second term signifies the interference from BS $x_1$ at $\ell_2$, and iii) third term corresponds to interference from all BSs $x \in  \Phi$ except $x_1$ and $x_2$ (equivalent to exclusion zone of $\mathcal{C}_1 \cup \mathcal{C}_2$). Step (c) follows by taking expectation over the fading links and applying Campbell's law to the interference in the third term, while the final step (d) follows by applying suitable limits of integration to the different integration regions.  The limits of the integration region $\mathcal{C}_1 \setminus \mathcal{C}_2$ depend on the three cases (disjoint, intersecting or engulfed) as defined in Remark 1, with its lower limit $a$ (See Fig. \ref{Fig:Scenarios}) taking values ${v-r_1,r_2}$ for cases 1 and  2 respectively. The integration region is zero for case 3 as $\mathcal{C}_1 \setminus \mathcal{C}_2 = \phi$ ($\mathcal{C}_1$ is engulfed inside $\mathcal{C}_2$). The integral $\mathcal{B}_1(r_1,r_2,\gamma)$ is summarized below:
\begin{align}  \label{B1Integral}
\mathcal{B}_1(r_1,r_2,\gamma) &= \left\{
     \begin{array}{lr}
       \int\limits_{v-r_1}^{v+r_1} 2\cos^{-1}\big(\frac{r^2+v^2-r_1^2}{2rv}\big)F(r_1,r_2,r,\gamma) \: r\ud r ,  \qquad \text{case 1} \\
       \int\limits_{r_{2}}^{v+r_1} 2\cos^{-1}\big(\frac{r^2+v^2-r_1^2}{2rv}\big)F(r_1,r_2,r,\gamma) \: r\ud r ,   \qquad \text{case 2}\\
       0, \qquad \qquad \qquad \qquad \: \: \: \qquad \qquad  \qquad \qquad \: \: \: \: \: \: \text{case 3}\\
      \end{array}
   \right.
\end{align}
\bibliographystyle{IEEEtran}
\bibliography{spatio-temporal correlation_v5_arxiv}
\end{document}